\documentclass[runningheads]{llncs}

\newif\ifIEEE

\newif\ifaddappendix
\newif\iftoolprototype

\toolprototypetrue

\PassOptionsToPackage{hyphens}{url}
\usepackage{hyperref}
\usepackage{amsmath}
\usepackage[english]{babel}
\usepackage{latexsym}
\usepackage{amssymb}	
\usepackage{verbatim}   
\usepackage{ifthen}
\usepackage{xfrac}      
\usepackage{graphicx}   
\usepackage{tikz}
\usetikzlibrary{arrows,automata,positioning}
\usepackage{float}
\usepackage{wrapfig}
\usepackage{paralist}
\usepackage{array}
\usepackage{multirow}
\usepackage{rotating}
\usepackage{mathtools}
\usepackage{proof}
\usepackage{caption}
\input xy
\xyoption{all}
\usepackage{color}
\usepackage{thmtools,thm-restate}
\usepackage{pgfplots}

\usepackage{algorithm}
\usepackage{algorithmic}

\newcommand{\denseparagraph}[1]{\vspace*{-2.6mm}\paragraph{#1}}
\renewcommand{\denseparagraph}[1]{\paragraph{#1}}

\renewcommand{\emptyset}{\varnothing}

\newcommand{\Reals}{\ensuremath{\mathbb{R}}}

\newcommand{\N}{\ensuremath{\mathbb{N}}}

\newcommand{\paths}{\ensuremath{\mathsf{paths}}}
\newcommand{\traces}{\ensuremath{\mathsf{traces}}}
\newcommand{\trace}{\ensuremath{\mathsf{trace}}}

\newcommand{\NOx}{\ensuremath{\mathrm{NO}_x}}

\newcommand{\Inputs}{{\mathsf{In}}}
\newcommand{\Outputs}{{\mathsf{Out}}}
\newcommand{\qOutputs}{{\ensuremath{\mathsf{Out}_\quiescence}}}

\newcommand{\inp}{\mathsf{i}}
\newcommand{\outp}{\mathsf{o}}

\newcommand{\Contract}{{\ensuremath{\mathcal{C}}}}

\newcommand{\Norm}{{\mathsf{StdIn}}}
\newcommand{\Std}{{\ensuremath{\mathcal{S}}}}
\newcommand{\qStd}{{\ensuremath{\mathcal{S}_\quiescence}}}
\newcommand{\StdPlus}{{\Std_{+}}}
\newcommand{\mStdPlus}{{\qStd}}

\newcommand{\Spec}{{\ensuremath{\mathcal{R}}}}

\newcommand{\STraces}{{\mathsf{STraces}}}

\newcommand{\BadInputs}{\mathit{BadTraces}}
\newcommand{\BadInputsFin}{\mathit{BadTracesFin}}

\newcommand{\inpbound}{\kappa_\inp}
\newcommand{\outpbound}{\kappa_\outp}

\newcommand{\robustly}{robustly}

\newcommand{\LTS}{\mathcal{L}}

\newcommand{\SUT}{\ensuremath{\mathcal{I}}}
\newcommand{\NoInp}{\ensuremath{\text{--}_i}}
\newcommand{\NoOutp}{\ensuremath{\text{--}_o}} 
\newcommand{\quiescence}{\ensuremath{\delta}}

\newcommand{\test}{t}

\newcommand{\out}{\mathsf{\sf out}}
\newcommand{\after}{\mathrel{\mathsf{after}}}
\newcommand{\ioco}{\ensuremath{\mathrel{\text{\rm\bf ioco}}}}

\newcommand{\pass}{\ensuremath{\text{\rm\bf pass}}}
\newcommand{\fail}{\ensuremath{\text{\rm\bf fail}}}
\newcommand{\passes}{\ensuremath{\text{ \rm\bf passes }}}

\newcommand{\TG}{\ensuremath{\mathsf{TG}}}
\newcommand{\DTG}{\ensuremath{\mathsf{DTG}}}
\newcommand{\DT}{\ensuremath{\mathsf{DT}}}

\newcommand{\nsum}{\mathop{\textstyle\sum}}

\newcommand{\true}{\textsf{true}}



\newcommand{\abs}[1]{\lvert {#1} \rvert}
\newcommand{\last}{\mathop{\mathsf{last}}}

\newcommand{\mapInp}[1]{#1 \ensuremath{{\downarrow_i}}}
\newcommand{\mapOut}[1]{#1 \ensuremath{{\downarrow_o}}}

\newcommand{\clean}{\mathsf{acc}}

\newcommand{\lift}{[\theta/\delta]}

\newcommand{\NEDC}{{\sc{NEDC}}}
\newcommand{\PowerNEDC}{{\sc{PowerNEDC}}}
\newcommand{\SineNEDC}{{\sc{SineNEDC}}}

\ifIEEE
\newtheorem{definition}{Definition}{\bfseries}{\itshape}
{\bfseries}{\itshape}
\newtheorem{lemma}{Lemma}{\bfseries}{\itshape}
\newtheorem{corollary}{Corollary}{\bfseries}{\itshape}
{\bfseries}{\itshape}

{\itshape}{\rmfamily}
{\itshape}{\rmfamily}
\newtheorem{example}{Example}{\itshape}{\rmfamily}

\renewenvironment{proof}{\begin{IEEEproof}}{\end{IEEEproof}}
\newcommand{\qed}{}
\fi

\newif\ifcomments
\newif\ifshowremoved

\usepackage{color}

\definecolor{lightblue}{RGB}{210,210,225}
\definecolor{lightred}{RGB}{225,210,210}
\definecolor{lightgreen}{RGB}{210,225,210}
\definecolor{lightyellow}{RGB}{225,222,200}
\definecolor{lightpurple}{RGB}{225,210,225}
\definecolor{warningyellow}{RGB}{247, 245, 187}
\definecolor{darkergreen}{RGB}{0,64,0}
\definecolor{darkred}{RGB}{128,0,0}
\definecolor{darkblue}{RGB}{0,0,128}
\definecolor{darkgreen}{RGB}{0,128,0}
\definecolor{darkpurple}{RGB}{128,0,128}
\definecolor{warningorange}{RGB}{124, 81, 0}
\definecolor{eyecancerpink}{rgb}{1.0, 0.0, 1.0}
\definecolor{radiationyellow}{rgb}{0.8, 1.0, 0.0}

\newcommand{\colorpar}[3]{\colorbox{#1}{\parbox{#2}{#3}}}

\newcommand{\marginremark}[4]{\marginpar{\colorpar{#2}{\linewidth}{\color{#1}\tiny{[#3]~ #4}}}}

\makeatletter
\def\THICKhrulefill{\leavevmode \leaders \hrule height 5pt\hfill \kern \z@}
\makeatother
\newcommand{\highlightedremark}[4]{%
  \begin{center}\fcolorbox{#1}{#2}{%
  \begin{minipage}{.98\linewidth}\color{#1}%
  \textbf{\THICKhrulefill[ #3 ]\THICKhrulefill}%
  \par\noindent#4\end{minipage}}\end{center}%
}

\newcommand{\remarkHH}[1]{\ifcomments\marginremark{eyecancerpink}{radiationyellow}{HH}{#1}\fi}
\newcommand{\remarkPRD}[1]{\ifcomments\marginremark{darkred}{lightred}{PRD}{#1}\fi}
\newcommand{\remarkSB}[1]{\ifcomments\marginremark{darkergreen}{lightgreen}{SB}{#1}\fi}

\newcommand{\hrmkSB}[1]{\ifcomments\highlightedremark{darkergreen}{lightgreen}{SB}{#1}\fi}

\newcommand{\removed}[1]{\ifcomments\highlightedremark{warningorange}{warningyellow}{REMOVED CONTENT}{\ifshowremoved #1 \else Add ``showremovedtrue'' below ``commentstrue'' to display removed content. \fi}\fi}

\begin{document}
\title{Doping Tests for Cyber-Physical Systems%
  \thanks{This work is partly supported by the ERC Grant 695614 (POWVER) by the Deutsche Forschungsgemeinschaft (DFG, German Research Foundation) grant 389792660 as part of TRR~248, see \url{https://perspicuous-computing.science}, by the Saarbr\"ucken Graduate School of Computer Science, by the Sino-German CDZ project 1023 (CAP), by ANPCyT PICT-2017-3894 (RAFTSys), and by SeCyT-UNC 33620180100354CB (ARES).%
	}%
}%

\author{Sebastian Biewer\inst{1} \and
Pedro D'Argenio\inst{1,2,3} \and
Holger Hermanns\inst{1}}

\institute{
  Saarland University, Saarland Informatics Campus, Germany
  \and Universidad Nacional de C\'ordoba, FAMAF, Argentina
  \and CONICET, Argentina
}

\def\hh#1{\remarkHH{#1}}

\maketitle

\begin{abstract}
 The software running in embedded or cyber-physical systems (CPS) is typically of proprietary nature, so users do not know precisely what the systems they own are (in)capable of doing. Most malfunctionings of such systems are not intended by the manufacturer, but some are, which means these cannot be classified as bugs or security loopholes. The most prominent examples have become public in the diesel emissions scandal, where millions of cars were found to be equipped with software violating the law, altogether polluting the environment and putting human health at risk. The behaviour of the software embedded in these cars was intended by the manufacturer, but it was not in the interest of society, a  phenomenon that has been called \emph{software doping}. Doped software is significantly different from buggy or insecure software and hence it is not possible to use classical verification and testing techniques to discover and mitigate software doping.

The work presented in this paper builds on existing definitions of software doping and lays the theoretical foundations for conducting software doping tests, so as to enable attacking evil manufacturers. The complex nature of software doping makes it very hard to effectuate doping tests in practice. We explain the biggest challenges and provide efficient solutions to realise doping tests despite this complexity.
\end{abstract}

\section{Introduction}
Embedded and cyber-physical systems are becoming more and more widespread as part of our daily life. Printers, mobile phones, smart watches, smart home equipment, virtual assistants, drones and batteries are just a few examples. Modern cars are even composed of a multitude of such systems. These systems can have a huge impact on our lives, especially if they do not work as expected. As a result, numerous approaches exist to assure quality of a system. The classical and most common type of malfunctioning is what is widely called ``bug''. Usually, a bug is a very small mistake in the software or hardware that causes a behaviour that is not intended or expected. Other types of malfunctioning are caused by incorrect or wrongly interpreted sensor data, physical deficiencies of a component, or are simply radiation-induced.

Another interesting kind of malfunction (also from an ethical perspective~\cite{DBLP:conf/isola/Baum16}) arises if the expectation of how the system should behave is different for two (or more) parties. Examples for such scenarios are widespread in the context of personal data privacy, where product manufacturers and data protection agencies have notoriously different opinions about how a software is supposed to handle personal data. Another example is the usage of third-party cartridges in printers. Manufacturers and users do not agree on whether their printer should work with third-party cartridges (the user's opinion) or only with those sold by the manufacturer (the manufacturer's opinion). Lastly, an example that received very high media attention are emission cleaning systems in diesel cars. There are regulations for dangerous particles and gases like CO$_2$ and NO$_2$ defining how much of these substances are allowed to be emitted during car operation. Part of these regulations are emissions tests, precisely defined test cycles that a car has to undergo on a chassis dynamometer~\cite{nedc}. Car manufacturers have to obey to these regulations in order to get admission to sell a new car model. The central weakness of these regulations is that the relevant behaviour of the car is only a trickle of the possible behaviour on the road. Indeed, several manufacturers equipped their cars with defeat devices that recognise if the car is undergoing an official emissions test. During the test, the car obeys the regulation, but outside test conditions, the emissions extruded are often significantly higher than allowed. Generally speaking, the phenomena described above are considered as incorrect software behaviour by one party, but as intended software behaviour by the other party (usually the manufacturer). In the literature, such phenomena are called software doping~\cite{BartheDFH16:isola,DBLP:conf/esop/DArgenioBBFH17}. 

The difference between software doping and 
bugs
is threefold: (1) There is a disagreement of intentions about what the software should do. (2) While a bug is most often a small coding error, software doping can be present in a considerable portion of the implementation. (3) Bugs can potentially be detected during production by the manufacturer, whereas software doping needs to be uncovered after production, by the other party facing the final product.
Embedded software is typically proprietary, so (unless one finds a way to breach into the intellectual property~\cite{DBLP:conf/sp/ContagLPDLHS17}) it is only possible to detect software doping by observation of the behaviour of the product, i.e., by black-box testing.

This paper 
develops the foundations for black-box testing approaches geared towards uncovering doped software in concrete cases. We will start off from an established formal notion of robust cleanness (which is the negation of software doping)~\cite{DBLP:conf/esop/DArgenioBBFH17}. Essentially, the idea of robust cleanness is based on a succinct specification (called a ``contract'') capturing the intended behaviour of a system with respect to all inputs to the system. Inputs are considered to be user inputs or environmental inputs given by sensors. The contract is defined by input and output distances on standard system trajectories supplemented by input and output thresholds. Simply put, the input distance and threshold induce a tube around the standard inputs, and similar for outputs. For any input in the tube around some standard input the system must be able to react with an output that is in the tube around the output possible according to the standard.
\begin{example}
For a diesel car the standard trajectory is the behaviour exhibited during the official emissions test cycle. The input distance measures the deviation in car speed from the standard. The  input threshold is a small number larger than the acceptable error tolerance of the cycle limiting the inputs considered of interest. The output distance then is the difference between (the total amount of) \NOx\ extruded by the car facing inputs of interest and that extruded if on the standard test cycle. 
For cars with an active defeat device we expect to see a violation of the contract even for relatively large output thresholds. 
\end{example}
A cyber-physical system (CPS) is influenced by physical or chemical dynamics. Some of this can be observed by the sensors the CPS is equipped with, but some portion might remain unknown, making proper analysis difficult. Nondeterminsm is a powerful way of representing  such uncertainty faithfully, and indeed the notion of robust cleanness supports non-deterministic reactive systems~\cite{DBLP:conf/esop/DArgenioBBFH17}. Furthermore, the analysis needs to consider (at least) two trajectories simultaneously, namely the standard trajectory and another that stays within the input tube. In the presence of  nondeterminism it might even become necessary to consider infinitely many trajectories at the same time. Properties over multiple traces are called hyperproperties~\cite{ClarksonS08}. 
In this respect, expressing robust cleanness as a hyperproperty needs both $\forall$ and $\exists$ trajectory quantifiers.
Formulas containing only one type of quantifier can be analysed efficiently, e.g., using model-checking techniques, but checking properties with alternating  quantifiers is known to be very complex~\cite{ClarksonFKMRS14:post,FinkbeinerRS15:cav}. Even more, testing of such problems is in general not possible. 
Assume, for example, a property requiring for a (non-deterministic) system that for every input $i$, there exists the output $o=i$, i.e., one of the systems possible behaviours computes the identity function. For black-box systems with infinite input and output domains the property can neither be verified nor falsified through testing. In order to verify the property, it is necessary to iterate over the infinite input set. For falsification 
one must show
that for some $i$ the system can not produce $i$ as output. However, not observing an output in finitely many steps does not rule out that this output can be generated.
As a result, there is no prior work (we are aware of) that targets the automatic generation of test cases for hyperproperties, let alone robust cleanness. 

The contribution of this paper is three-fold. 
(1) We observe that standard behaviour, in particular when derived by common standardisation procedures, can be represented by finite models, and we identify under which conditions the resulting contracts are (un)satisfiable.
(2) For a given satisfiable contract we construct the largest non-deterministic model that is robustly clean w.r.t. this contract. We integrate this model into a model-based testing theory, 
which can provide a non-deterministic algorithm to derive sound test suites.
(3) We develop a testing algorithm for bounded-length tests and discretised input/output values. We present test cases for the diesel emissions scandal and execute these tests with a real car on a chassis dynamometer.

\section{Software Doping on Reactive Programs}
\label{sec:sd:react}

Embedded software is reactive, it reacts to inputs received from sensors by producing outputs that are meant to control the device functionality.   
We consider a reactive program as a function $P:\Inputs^\omega\to
2^{(\Outputs^\omega)}$ on infinite sequences of inputs so that the program reacts to the $k$-th input in the input sequence by producing non-deterministically the $k$-th output in each respective
output sequence.  Thus, the program can be seen, for instance, as a
(non-deterministic) Mealy or Moore machine.
Moreover, we consider an equivalence relation
${\approx}\subseteq{\Inputs^\omega\times\Inputs^\omega}$ that equates
sequences of inputs. To illustrate this, think of the program embedded in a printer.
Here $\approx$ would for instance equate input sequences that agree
with respect to submitting the same documents regardless of the cartridge brand,
the level of the toner (as long as there is sufficient), etc.  We furthermore
consider the set $\Norm\subseteq\Inputs^\omega$ of inputs of interest
or \emph{standard inputs}. In the previous example, $\Norm$ contains
all the input sequences with compatible cartridges and printable
documents.
The definitions given below are simple adaptations of those given
in~\cite{DBLP:conf/esop/DArgenioBBFH17} (but where parameters are  instead treated as parts of the inputs).


\begin{definition}\label{def:clean:react}
  A reactive program $P$ is \emph{clean} if for all inputs
  $\inp,\inp'\in\Norm$ such that $\inp\approx\inp'$, $P(\inp) =
  P(\inp')$. Otherwise it is  \emph {doped}.
\end{definition}
This definition states that a program is \emph{clean} if its
execution exhibits the same visible sequence of output when supplied
with two equivalent inputs, provided such inputs comply with the given
standard $\Norm$.
Notice that the behaviour outside $\Norm$ is deemed immediately
clean since it is of no interest.

In the context of the printer example, a program that would fail to print a document when provided with an ink cartridge from a third-party manufacturer, but would otherwise succeed to print would be considered doped, since this difference in output behaviour is captured by the above definition. For this, the inputs (being pairs of document and printer cartridge) must be considered equivalent (not identical), which comes down to ink cartridges being compatible.

However, the above definition is not very helpful for cases that need to
preserve certain intended behaviour \emph{outside} of the standard
inputs $\Norm$.  This is clearly the case in the diesel emissions
scandal where the standard inputs are given precisely by the emissions test, but the behaviour observed there is assumed to generalise beyond the singularity of this test setup. It is meant to ensure that the amount of NO$_2$ and NO (abbreviated as \NOx) in the car exhaust gas does not deviate considerably \emph{in general}, and comes with a legal prohibition of defeat mechanisms that simply turn off the cleaning mechanism. This legal framework is obviously a bit short sighted, since it can be circumvented by mechanisms that alter the behaviour gradually in a continuous manner, but in effect drastically. 
In a nutshell, one expects that if the input values observed by the electronic control unit (ECU) of a diesel vehicle deviate within ``reasonable distance'' from the \emph{standard} input values provided during the lab emission test, the amount of \NOx\ found in the exhaust gas is still within the regulated threshold, or at least
it does not exceed it more than a ``reasonable amount''.

This motivates the need to introduce the notion of distances on inputs and
outputs.  More precisely, we consider distances on finite traces:
$d_\Inputs:(\Inputs^*\times\Inputs^*)\to\Reals_{\geq0}$ and
$d_\Outputs:(\Outputs^*\times\Outputs^*)\to\Reals_{\geq0}$. 
Such distances are required to be 
pseudometrics. ($d$ is a pseudometric
if $d(x,x)=0$, $d(x,y)=d(y,x)$ and $d(x,y)\leq d(x,z)+d(z,y)$ for all
$x$, $y$, and $z$.)
With this, D'Argenio et. al~\cite{DBLP:conf/esop/DArgenioBBFH17} provide a definition of robust cleanness that considers two parameters:
parameter $\inpbound$ refers to the acceptable distance an input may
deviate from the norm to be still considered, and parameter
$\outpbound$ that tells how far apart outputs are allowed to be in
case their respective inputs are within $\inpbound$ distance (Def.~\ref{def:ed-clean:react}
spells out the Hausdorff distance used in~\cite{DBLP:conf/esop/DArgenioBBFH17}).

\begin{definition}\label{def:ed-clean:react}
  Let $\sigma[..k]$ denote the $k$-th prefix of the sequence $\sigma$.
  A reactive program $P$ is \emph{\robustly\ clean} if for all input
  sequences $\inp,\inp'\in\Inputs^\omega$ with $\inp\in\Norm$,
  for all $k\geq 0$ such that
  $d_\Inputs(\inp[..j],\inp'[..j])\leq\inpbound$ for all $j\leq k$,
  the following holds:
%
%
  \begin{enumerate}[(1)]
  \item\label{def:ed-clean:seq:i}
    for all $\outp\in P(\inp)$ there exists $\outp'\in P(\inp')$ such
    that $d_\Outputs(\outp,\outp') \leq \outpbound$, and
  \item\label{def:ed-clean:seq:ii}
    for all $\outp'\in P(\inp')$ there exists $\outp\in P(\inp)$ such
    that $d_\Outputs(\outp,\outp') \leq \outpbound$,
  \end{enumerate}
  where $P(\inp)[..k]=\{\outp[..k] \mid \outp\in P(\inp)\}$
  and similarly for $P(\inp')[..k]$.
\end{definition}
Notice that this is what we actually need for the non-deterministic
case: each output of one of the program instances should be matched
within ``reasonable distance'' by some output of the other program
instance.
Also notice that $\inp'$ does not need to satisfy $\Norm$, but it will
be considered as long as it is within $\inpbound$ distance of any
input satisfying $\Norm$.  In such a case, outputs generated by
$P(\inp')$ will be requested to be within $\outpbound$ distance of
some output generated by the respective execution induced by a
standard input.

We remark that Def.~\ref{def:ed-clean:react} entails the existence of
a \emph{contract} which defines the set of standard inputs $\Norm$, the
tolerance parameters $\inpbound$ and $\outpbound$ as well as the
distances $d_\Inputs$ and $d_\Outputs$. In the context of diesel engines, one might imagine that the values to be considered, especially the tolerance parameters  $\inpbound$ and $\outpbound$ for a particular car model are made publicly available (or are even advertised by the car manufacturer), so as to enable potential customers to discriminate between different car models according to the robustness they reach in being clean. It is also imaginable that the tolerances and distances are fixed by the legal authorities as part of environmental regulations.

\section{Robustly Clean Labelled Transition Systems}
This section develops the framework needed for an effective theory of black-box doping tests based on the above concepts. In this, the standard behaviour (e.g. as defined by the emission tests) and the robust cleanness definitions together will induce a set of reference behaviours that then serve as a model in a model-based conformance testing approach. To set the stage for this, we recall the definitions of labelled transition systems (LTS) and input-output transitions systems (IOTS) together with 
Tretmans' notion on model-based conformance testing~\cite{DBLP:phd/basesearch/Tretmans92}. We then recast
the characterisation of robust cleanness (Def.~\ref{def:ed-clean:react}) in terms of LTS.

\begin{definition} \label{def:LTS}\label{def:IOTS}
  A \emph{labelled transition system (LTS)} with inputs and outputs is
  a tuple $\langle Q, \Inputs, \Outputs, {\rightarrow}, q_0 \rangle$ where
  \begin{inparaenum}[(i)]
  \item $Q$ is a (possibly uncountable) non-empty set of states;
  \item $L = \Inputs \uplus \Outputs$ is a (possibly uncountable) set of labels;
  \item ${\rightarrow} \subseteq {Q \times 
    L
    \times Q}$ 
    is the transition relation;
  \item $q_0 \in Q$ is the initial state.
  \end{inparaenum}
  We say that a LTS is an \emph{input-output transition system
    (IOTS)} if it is input-enabled in any state, i.e.,
  for all $s\in Q$ and $a\in\Inputs$ there is some $s'\in Q$ such that
  $s\xrightarrow{a}s'$.
\end{definition}
For ease of presentation, we do not consider internal transitions. The following definitions will be  used throughout the paper.  A
\emph{finite path} $p$ in an LTS $\LTS$ is a sequence
$s_1a_1s_2a_2\ldots a_{n-1}s_n$
with $s_i \xrightarrow{a_i } s_{i+1}$ for all $1 \leq i < n$.
Similarly, an \emph{infinite path} $p$ in $\LTS$ is a sequence
%
$s_1a_1s_2a_2\ldots$
with $s_i \xrightarrow{a_i } s_{i+1}$ for all $i \in \N$.
Let $\paths_*(\LTS)$ and $\paths_\omega(\LTS)$ be the sets of all finite and
infinite paths of $\LTS$ beginning in the initial states, respectively. 
The sequence $a_1 a_2 \cdots a_n$ is a \emph{finite trace} of $\LTS$
if there is a finite path
$s_1a_1s_2a_2\ldots a_{n}s_{n+1} \in \paths_*(\LTS)$,
and $a_1 a_2 \cdots$ is an \emph{infinite trace} if there is an
infinite path
$s_1a_1s_2a_2\ldots\in \paths_\omega(\LTS)$.
If $p$ is a path, we let $\trace(p)$ denote the trace defined by $p$.
Let $\traces_*(\LTS)$ and $\traces_\omega(\LTS)$ be the sets of all
finite and infinite traces of $\LTS$, respectively.
We will use $\LTS_1 \subseteq \LTS_2$ to denote that $\traces_\omega(\LTS_1) \subseteq \traces_\omega(\LTS_2)$.

\denseparagraph{Model-Based Conformance Tests.}

In the following we recall the basic notions of \ioco\ conformance
testing~\cite{DBLP:phd/basesearch/Tretmans92,DBLP:journals/cn/Tretmans96,DBLP:conf/fortest/Tretmans08}, and refer to the mentioned literature for more details.
In this setting, it is  assumed that the implemented system under test (IUT) 
$\SUT$ can be modelled as an IOTS while the specification of
the required behaviour is given in terms of a LTS $\mathcal{S}$.  The
idea of whether the IUT $\SUT$ \emph{conforms to} the
specification $\mathcal{S}$ is formalized by means of the \ioco\ relation
which we define in the following.

We first need to identify the \emph{quiescent} (or \emph{suspended})
states.  A state is quiescent whenever it cannot proceed autonomously,
i.e., it cannot produce an output. We will make each such state identifiable by
adding a quiescence transition to it, in the form of a loop with the distinct label $\delta$.

\begin{definition} \label{def:delta-closure}
  Let $\LTS = \langle Q, \Inputs, \Outputs, {\rightarrow}, q_0 \rangle$ be an LTS.
  The \emph{quiescence closure} (or \emph{$\quiescence$-closure}) of
  $\LTS$ is the LTS
  $\LTS_\quiescence \coloneqq \langle Q, \Inputs, \Outputs \cup \{\quiescence\}, {\rightarrow_\quiescence}, q_0 \rangle$
  with
  ${\rightarrow_\quiescence} \coloneqq  {\rightarrow}  \cup  \{{s \xrightarrow{\;\quiescence\;}_\quiescence s} \mid \forall o \in \Outputs, t \in Q: {s \mathrel{\ \not\!\!\!\xrightarrow{\;o\;}} t} \}$.
  Using this we define the \emph{suspension traces} of $\LTS$ by
  $\traces_*(\LTS_\quiescence)$.
\end{definition}
Let $\LTS$ be an LTS with initial state $q_0$ and
$\sigma=a_1\,a_2\ldots a_n\in\traces_*(\LTS)$.
We define $\LTS\after\sigma$ as the set
%
$\{q_n \mid q_0 a_1 q_1 a_2 \ldots a_n q_n\in\paths_*(\LTS)\}$.
For a state $q$, let
$\out(q)=\{o\in\Outputs\cup\{\delta\}\mid\exists q': q\xrightarrow{o}q'\}$
and for a set of states $Q'\subseteq Q$, let
$\out(Q')=\bigcup_{q\in Q'}\out(q)$.

The idea behind the \ioco\ relation is that any output produced by the
IUT must have been foreseen by its specification, and moreover, any input in
the IUT not foreseen in the specification may introduce new
functionality. \ioco\ captures this by harvesting concepts from refusal
testing.
As a result, $\SUT\ioco \textit{Spec}$ is defined to hold whenever
$\out(\SUT_\quiescence \after \sigma) \subseteq \out(\textit{Spec}_\quiescence \after \sigma)$
for all $\sigma\in\traces_*(\textit{Spec}_\quiescence)$.

The base principle of \emph{conformance testing} now is to assess by means of testing
whether the IUT conforms to its specification w.r.t.\ \ioco. An algorithm to derive a  corresponding test suite $T_{\textit{Spec}}$ is available~\cite{DBLP:journals/cn/Tretmans96,DBLP:conf/fortest/Tretmans08}, so that for any
IUT $\SUT$, $\SUT\ioco\textit{Spec}$ iff
$\SUT$ passes all tests in $T_{\textit{Spec}}$. 

It is important to remark that the specification in the  setting considered here is missing. Instead, we need to construct the specification from the standard inputs and the respective observed outputs, together with the distances and the tresholds given by the contract.  Furthermore, this needs to respect the ${\forall}-{\exists}$ interaction required by the cleanness property (Def.~\ref{def:ed-clean:react}).

\denseparagraph{Software Doping on LTS.}

To capture the notion of software
doping in the context of LTS, we
provide two projections of a trace,  projecting to a sequence of
the appearing inputs, respectively outputs.  To do this, we extend
the set of labels by adding the input $\NoInp$, that indicates that in
the respective step some output (or quiescence) was produced (but masking the precise output), and the output $\NoOutp$
that indicates that in this step some (masked) input was given.

The \emph{projection on inputs} $\mapInp{} : L^\omega \rightarrow
(\Inputs \cup \{\NoInp\})^\omega$ and the \emph{projection on outputs}
$\mapOut{} : L^\omega \rightarrow (\Outputs \cup \{\NoOutp\})^\omega$
are defined for all traces $\sigma$ and $k\in\N$ as follows:
 $\mapInp{\sigma} [k] \coloneqq
  \textbf{ if } \sigma[k] \in \Inputs \textbf{ then } \sigma[k]
  \textbf{ else } \NoInp$ and
 $\mapOut{\sigma} [k] \coloneqq
  \textbf{ if } \sigma[k] \in \Outputs \textbf{ then } \sigma[k]
  \textbf{ else } \NoOutp$.
They are lifted to sets of traces in the usual elementwise way.

\begin{definition} \label{def:LTS:Std}
  A LTS $\Std$ is \emph{standard} for a LTS $\LTS$, if for all
  $\sigma \in \traces_\omega(\Std)$ and $\sigma'\in
  \traces_\omega(\LTS)$, $\mapInp{\sigma} = \mapInp{\sigma'}$ implies
  $\sigma' \in \traces_\omega(\Std)$.
\end{definition}
The above definition provides our LTS-specific interpretation of the notion of $\Norm$ for a given program $P$ modelled in terms of LTS $\LTS$. $\Norm$ is implicitly determined as the input sequences $\mapInp{\traces_\omega(\Std)}$ occurring in $\Std$, which contains both the standard inputs and the outputs that are produced in $\LTS$, and altogether covers exactly the traces of $\LTS$ whose input sequences are in $\Norm$.
If instead  $\LTS$ and $\Norm\subseteq (\Inputs \cup \NoInp)^\omega$ are given, LTS $\Std$ can be defined such that
$\sigma\in\traces_\omega(\Std)$ iff $\mapInp{\sigma}\in\Norm$ and
$\sigma\in\traces_\omega(\LTS)$.  This $\Std$ is indeed
standard for $\LTS$ and contains exactly all traces of $\LTS$ whose
input sequences are in $\Norm$. 

In this new setting, we assume that the distance functions $d_\Inputs$
and $d_\Outputs$ run on traces containing labels $\NoInp$ and
$\NoOutp$, i.e. they are pseudometrics in
$(\Inputs\cup\{\NoInp\})^* \times (\Inputs\cup\{\NoInp\})^* \rightarrow \Reals_{\geq0}$
and
$(\Outputs\cup\{\NoOutp\})^* \times (\Outputs\cup\{\NoOutp\})^* \rightarrow \Reals_{\geq0}$,
respectively.

Now the definition of \robustly\ clean can be restated in terms of
LTS as follows.

\begin{definition}\label{def:ed-clean:LTS}
  Let $\LTS$ be an IOTS and $\Std$ a standard LTS.  $\LTS$ is
  \emph{\robustly\ clean} if for all
  $\sigma,\sigma'\in\traces_\omega(\LTS_\quiescence)$ such that
  $\sigma\in\traces_\omega(\Std_\quiescence)$ then for all $k\geq 0$ such
  that
  $d_\Inputs(\mapInp{\sigma[..j]},\mapInp{\sigma'[..j]})\leq\inpbound$
  for all $j\leq k$, the following holds:
  \begin{enumerate}[\hspace{0pt}1.\hspace{-3pt}]
  \item\label{def:ed-clean:LTS:i}%
  \mbox{there exists $\sigma''\!\in\traces_\omega(\LTS_\quiescence)$ s.t. 
    $\mapInp{\sigma'\!} = \mapInp{\sigma''\!}$ and
    $d_\qOutputs(\mapOut{\sigma[..k]},\mapOut{\sigma''[..k]})\leq\outpbound$\phantom{.}  }
  \item\label{def:ed-clean:LTS:ii}%
    \mbox{there exists $\sigma''\!\in\traces_\omega(\LTS_\quiescence)$ s.t. 
    $\mapInp{\sigma\!}=\mapInp{\sigma''\!}$ and
    $d_\qOutputs(\mapOut{\sigma'[..k]},\mapOut{\sigma''[..k]})\leq\outpbound$. }
  \end{enumerate}
\end{definition}
Following the principles of model-based testing, Def.~\ref{def:ed-clean:LTS} takes specific care of quiescence in a system. In order to properly consider quiescence in the context of robust cleanness it must be considered as a unique output. As a consequence, in the presence of  a contract $\Contract = \langle \Inputs, \Outputs, \Std, d_\Inputs, d_\Outputs, \inpbound, \outpbound \rangle$, we use -- instead of $\Std$, $\Outputs$ and $d_\Outputs$ -- the quiescence closure $\Std_\quiescence$ of $\Std$, $\qOutputs = \Outputs \cup \{ \quiescence \}$ and an extended output distance defined as
$d_{\Outputs_\quiescence}(\sigma_1, \sigma_2) \coloneqq d_\Outputs({\sigma_1}_{{\backslash}\quiescence}, {\sigma_2}_{{\backslash}\quiescence})$ if 
${\sigma_1[i]=\quiescence} \Leftrightarrow {\sigma_2[i]=\quiescence}$ for all $i$, 
and $d_{\Outputs_\quiescence}(\sigma_1, \sigma_2) \coloneqq \infty$ otherwise, where $\sigma_{{\backslash}\quiescence}$ is the same as $\sigma$ whith all $\quiescence$ removed.

Def.~\ref{def:ed-clean:LTS} echoes the semantics of
the HyperLTL  interpretation appearing in Proposition 19
of~\cite{DBLP:conf/esop/DArgenioBBFH17} restricted to programs with no
parameters.  Thus, the proof showing that Def.~\ref{def:ed-clean:LTS}
is the correct interpretation of Def.~\ref{def:ed-clean:react} in
terms of LTS, can be obtained  in a way  similar to that of Prop.~19
in~\cite{DBLP:conf/esop/DArgenioBBFH17}.

\section{Reference Implementation for Contracts}

As mentioned before, doping tests need to be based on a contract \Contract, which we assume given. \Contract\ specifies the domains $\Inputs$, $\Outputs$, a standard LTS $\Std$, the distances $d_\Inputs$ and $d_\Outputs$ and the bounds $\inpbound$ and $\outpbound$. 
%
%
We intuitively
expect the contract to be satisfiable in the sense that it never enforces a single input sequence of the implementation to keep outputs close enough to two different executions of the specification while their outputs stretch too far apart.   
We show such a problematic case in the following example.

%
\denseparagraph{Example 2.}
\begin{wrapfigure}[12]{r}{0.18\textwidth}
  \vspace*{-9mm}
  \hspace*{-3mm}
  \begin{tikzpicture}[
      ->,>=stealth',shorten >=1pt,
      state/.style={draw=black, circle},
      auto,
    ]
    \node[state] (s0)               {};
    \node        (g2) [below=6mm of s0] {};
    \node[state] (s1) [left=5mm of g2]  {};
    \node[state] (s3) [right=5mm of g2] {};
  
    \node[state] (s6) [below of=s3] {};

      \node [left of=s0, xshift=2mm, yshift=1mm] {$\qStd$};
    
    \path (s0) edge node[midway,left]  {$\scriptstyle i{-}\inpbound\ $} (s1)
               edge node[midway,right] {$\scriptstyle \ i{+}\inpbound$} (s3)
          (s3) edge node[midway,right] {$\scriptstyle o$} (s6)
          (s0) edge [loop right, draw=black] node[midway, text=black] {$\scriptstyle \quiescence$} (s0)
          (s6) edge [loop left, draw=black] node[midway, text=black] {$\scriptstyle \quiescence$} (s6)
          (s1) edge [loop below, draw=black] node[midway, text=black] {$\scriptstyle \quiescence$} (s1);
          
  \end{tikzpicture}
  
\vspace{0.1cm}

\hspace*{-7.5mm}
  \begin{tikzpicture}[
      ->,>=stealth',shorten >=1pt,
      state/.style={draw=black, circle},
      auto,
    ]
    \node[state] (s0)               {};
    \node        (g2) [below=6mm of s0] {};
    \node[state] (s1) [left=5mm of g2]  {};
    \node[state] (s3) [right=5mm of g2] {};
  
    \node [state] (s6) [below of=s3] {};
    \node [left of=s0, xshift=-.65cm, yshift=.2cm] {};

 \node [left of=s0, xshift=2mm, yshift=1mm] {$\LTS$};

    \node[state] (t1) [below=.55 of s0] {};
    \node (t2) [below of = t1] {};

    \path (s0) edge node[midway,left]  {$\scriptstyle i{-}\inpbound\ $} (s1)
               edge node[midway,right] {$\scriptstyle \ i{+}\inpbound$} (s3)
          (s3) edge node[midway,right] {$\scriptstyle o$} (s6);
          
     \path (s0) edge node[midway,right] {$\scriptstyle i$} (t1)
     		  (t1) edge node[midway,right] {$\scriptstyle x$} (t2);
  \end{tikzpicture}
    
\end{wrapfigure}
On the right a quiescence-closed standard LTS $\qStd$ for an implementation $\LTS$ (shown below) is  depicted. For simplicity 
some input transitions are omitted. 
Assume $\Outputs=\{o\}$ and $\Inputs = \{i, i-\inpbound, i+\inpbound\}$. Consider the transition labelled $x$ of~$\LTS$. This must be one of either $o$ or $\quiescence$, but we will see that either choice leads to a contradiction w.r.t. the output distances induced. 
The input projection of the middle path in $\LTS$ is  $i\;\NoInp$ and the input distance to $(i-\inpbound)\;\NoInp$ and $(i+\inpbound)\;\NoInp$ is exactly $\inpbound$, so both branches $(i+\inpbound)\;o$ and $(i-\inpbound)\;\quiescence$ of \qStd{} must be considered to determine $x$.  
For $x=o$, the output distance of $\NoOutp\;x$ to $\NoOutp\;o$ in the right branch of \qStd{} is $0$, i.e. less than $\outpbound$. However, $d_{\qOutputs}(\NoOutp\;\quiescence, \NoOutp\;o) = \infty > \outpbound$. Thus the output distance to the left branch of \qStd{} is too high  if picking $o$. 
Instead picking $x=\quiescence$ does not work either, for the symmetric reasons, the problem switches sides. 
Thus, neither picking $o$ nor $\quiescence$ for $x$ satisfies robust cleanness here. Indeed, no implementation satisfying robust cleanness exists  for the given contract. 

We would expect that a correct implementation fully entails the standard behaviour. So, to satisfy a contract, the standard behaviour itself must be robustly clean.
This and the need for satisfiability of particular inputs lead to Def.~\ref{def:satisfiableContract}.

\begin{definition}[Satisfiable Contract]
\label{def:satisfiableContract}
Let $\Inputs, \Outputs, \Std, d_\Inputs, d_\Outputs, \inpbound$ and $\outpbound$ define some contract $\Contract$. Let \emph{input} $\sigma_i \in (\Inputs \cup \{ \NoInp \})^\omega$ be the input projection of some trace. $\sigma_i$ is \emph{satisfiable} for \Contract\ if and only if  for every standard trace $\sigma_S \in \traces_\omega(\Std_\quiescence)$ and $k > 0$ such that for all $j \leq k$ $d_\Inputs(\mapInp{\sigma_i[..j], \sigma_S[..j]}) \leq \inpbound$ there is some implementation $\LTS$ that satisfies Def.~\ref{def:ed-clean:LTS}.\ref{def:ed-clean:LTS:ii} w.r.t. \Contract\ and has some trace $\sigma \in \traces_\omega(\LTS_\quiescence)$ with $\mapInp{\sigma} = \sigma_i$ and  $d_\qOutputs(\mapOut{\sigma[k]}, \mapOut{\sigma_S[k]}) \leq \outpbound$.

\Contract\ is \emph{satisfiable} if and only if all inputs $\sigma_i \in (\Inputs \cup \{ \NoInp \})^\omega$ are satisfiable for \Contract\ and if $\Std_\quiescence$ is robustly clean w.r.t. contract $\Contract_\Std = \langle \Norm, \Outputs, \Std, d_\Inputs, d_\Outputs,
  \inpbound, \outpbound \rangle$.
A contract that is not satisfiable is called \emph{unsatisfiable}.
\end{definition}

Given a satisfiable contract it is always possible to construct an implementation that is robustly clean w.r.t. to this contract. Furthermore, for every contract there is exactly one implementation (modulo trace equivalence)
that contains all possible outputs that satisfy robust cleanness. Such an implementation is called the \emph{largest implementation}.

\begin{definition}[Largest Implementation]
Let \Contract\ be a contract and $\LTS$ an implementation that is robustly clean w.r.t. \Contract. $\LTS$ is the \emph{largest implementation within \Contract} if and only if for every $\LTS'$ that is robustly clean w.r.t. \Contract\ it holds that $\traces_\omega(\LTS'_\quiescence) \subseteq \traces_\omega(\LTS_\quiescence)$.
\end{definition}

In the following, we will focus on the fragment of satisfiable contracts with standard behaviour defined by finite LTS. For unsatisfiable contracts, testing is not necessary, because every implementation is not robustly clean w.r.t. to $\Contract$. Finiteness of $\Std$ will be necessary to make testing feasible in practice. 
For simplicity we will further assume
\emph{past-forgetful} output distance functions. That is, 
$d_\Outputs(\sigma_1,\sigma_2)=d_\Outputs(\sigma'_1,\sigma'_2)$
whenever $\last(\sigma_1)=\last(\sigma'_1)$ and
$\last(\sigma_2)=\last(\sigma'_2)$ (where $\last(a_1\ a_2\ldots a_n)=a_n$.)
Thus, we simply assume that
$d_\Outputs:(\Outputs{\cup}\{\NoOutp\}\times\Outputs{\cup}\{\NoOutp\})\to\Reals_{\geq0}$, i.e., the output distances are determined by the last output only. 
We remark that
$d_\qOutputs(\quiescence,o)=\infty$
for all $o\neq\quiescence$.

We will now show how to construct the largest implementation for any contract (of the fragment we consider), which we name \emph{reference implementation}~\Spec.  
%
It 
is derived from $\mStdPlus$ by adding inputs and outputs in such a way that whenever  the input sequence  leading to a particular state is within $\inpbound$ distance of an input sequence
$\sigma_i$ of $\mStdPlus$, then the outputs possible in such a state
should be at most $\outpbound$ distant from those outputs possible in the unique
state on $\mStdPlus$ reached through $\sigma_i$.  This ensures that
$\Spec$ will satisfy condition 2) in Def.~\ref{def:ed-clean:LTS}.
\denseparagraph{Reference implementation.}
To construct the reference implementation $\Spec$ we decide to model the quiescence transitions explicitly instead of using the quiescence closure. We
preserve the property, that in each state of the LTS it is possible to
do an output or a quiescence transition. 
The construction of $\Spec$ proceeds by adding all transitions
that satisfy the second condition of Def.~\ref{def:ed-clean:LTS}.

\begin{definition}\label{def:testing:approxSpec}
  Given a 
  standard LTS
  $\mStdPlus = \langle Q, \Inputs, \Outputs, {\rightarrow_\mStdPlus}, \epsilon \rangle$,
  bounds $\inpbound$ and $\outpbound$, and distances
  $d_\Inputs$ and $d_\Outputs$, 
  the \emph{reference implementation} $\Spec$ is the LTS
  $\langle Q, \Inputs, \Outputs, {\rightarrow_{\Spec}}, \epsilon \rangle$
  where $\rightarrow_{\Spec}$ is defined by
  \[
    \infer{
      \sigma \xrightarrow{a}_{\Spec} \sigma \cdot a
    }{\begin{array}{l}
        \forall \sigma_i \in \mapInp{\traces_\omega(\mStdPlus)}:\\[.4ex]
        \quad
        (\forall j \leq |\sigma|+1:
        d_\Inputs(\mapInp{(\sigma \cdot a)}[..j], \sigma_i[..j]) \leq \inpbound)\\[.4ex]
        \qquad \Rightarrow
        \exists \sigma_S\in\traces_\omega(\mStdPlus):
        \mapInp{\sigma_S} = \sigma_i \land d_\qOutputs(\mapOut{a}, \mapOut{\sigma_S[|\sigma|+1]}) \leq \outpbound
      \end{array}
    }
  \]
\end{definition}
Notably, $\Spec$ is deterministic, since only
transitions of the form
$\sigma \xrightarrow{a}_{\Spec} \sigma \cdot a$
are added. 
As a consequence of this determinism, outputs and quiescence may coexists as options in a state, i.e. they are not mutually exclusive.

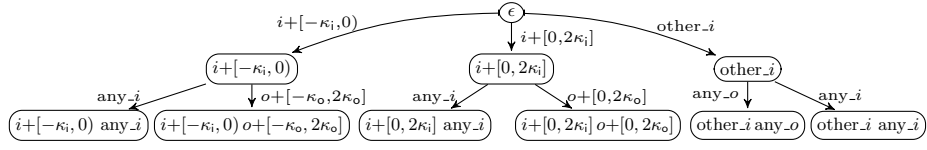
\begin{figure*}[t]
  \centering
    \begin{tikzpicture}[
        ->,>=stealth',shorten >=1pt,
        scale=0.75,
        state/.style={rectangle,draw,rounded corners=1.7mm,transform shape},
        auto,
      ]
      \node[state] (s0)               {$\,\epsilon\,$};
      \node[state] (s2) [below of=s0] {$i{+}[0,2\inpbound]$};
      \node        (g1) [left=.47cm of s2]  {};
      \node        (g3) [right=.47cm of s2] {};
      \node[state] (s1) [left=2.1cm of g1]  {$i{+}[-\inpbound,0)$};
      \node[state] (s3) [right=1.9cm of g3] {$\text{other}\_i$};
      \node[state] (s5) [below of=s1] {$i{+}[-\inpbound,0)\,o{+}[-\outpbound,2\outpbound]$};
      \node[state] (s4) [left=.1cm of s5]  {$i{+}[-\inpbound,0)\ \text{any}\_i$};
      \node[state] (s6) [below of=g1] {$i{+}[0,2\inpbound]\ \text{any}\_i$};
      \node[state] (s7) [below of=g3] {$i{+}[0,2\inpbound]\,o{+}[0,2\outpbound]$};
      \node[state] (s8) [below of=s3] {$\text{other}\_i\,\text{any}\_o$};
      \node[state] (s9) [right=.1cm of s8] {$\text{other}\_i\ \text{any}\_i$};

      \path (s0) edge[bend right=10] node[midway,left]  {$\scriptstyle i{+}[-\inpbound,0)\quad$}            (s1)
                 edge node[midway,right] {$\scriptstyle i{+}[0,2\inpbound]$}            (s2)
                 edge[bend left=10] node[midway,right] {$\scriptstyle \quad\text{other}\_i$}               (s3)
            (s1) edge node[midway,left]  {$\scriptstyle \text{any}\_i\ \ $}                 (s4)
                 edge node[midway,right] {$\scriptstyle o{+}[-\outpbound,2\outpbound]$} (s5)
            (s2) edge node[midway,left]  {$\scriptstyle \text{any}\_i\,$}                 (s6)
                 edge node[midway,right] {$\scriptstyle \,o{+}[0,2\outpbound]$}           (s7)
            (s3) edge node[midway,left]  {$\scriptstyle \text{any}\_o$}                 (s8)
                 edge node[midway,right] {$\scriptstyle \,\text{any}\_i$}                 (s9);
    \end{tikzpicture}
    
    \caption{The reference implementation $\Spec$ of $\Std$ in 
    Example~3.}\label{fig:ex:stdplusoverlineref}
\end{figure*}

\begin{wrapfigure}[7]{r}{0.27\textwidth}
\vspace{-0.73cm}
    \begin{tikzpicture}[
        ->,>=stealth',shorten >=1pt,
        state/.style={rectangle,draw,rounded corners=1.7mm},
        auto,
      ]
      \node[state] (s0)               {$s_0$};
      \node[state] (s2) [below of=s0] {$s_2$};
      \node[state] (s1) [left of=s2]  {$s_1$};
      \node[state] (s3) [right of=s2] {$s_3$};
      \node[state] (s4) [below of=s1] {$s_4$};
      \node[state] (s5) [below of=s2] {$s_5$};
      \node[state] (s6) [below of=s3] {$s_6$};
      \node [left of=s0, xshift=-.2cm, yshift=-2mm] {$\Std$};

      \path (s0) edge node[midway,left]  {$\scriptstyle i$}              (s1)
                 edge node[midway,left]  {$\scriptstyle i$}              (s2)
                 edge node[midway,right] {$\scriptstyle i{+}\inpbound$}  (s3)
            (s1) edge node[midway,left]  {$\scriptstyle o$}              (s4)
            (s2) edge node[midway,left]  {$\scriptstyle o{+}\outpbound$} (s5)
            (s3) edge node[midway,right] {$\scriptstyle o{+}\outpbound$} (s6);
    \end{tikzpicture}
\end{wrapfigure}
\denseparagraph{Example 3.} 
  Fig.~\ref{fig:ex:stdplusoverlineref} gives a schematic representation
  of the reference implementation $\Spec$ for the LTS 
  $\Std$ 
  on the right.
Input (output) actions are
  denoted with letter $i$ ($o$, respectively), quiescence transitions are omitted. 
  We use
  Euclidean distances $\parallel \!\! \cdot \!\! \parallel$, so that
  $d_\Inputs(i, i') \coloneqq \; \parallel \!\! i-i' \!\! \parallel$ and $d_\Outputs(o, o') \coloneqq \; \parallel \!\! o-o' \!\! \parallel$.
  For this example, the quiescence closure $\Std_\quiescence$ looks like $\Std$ but with
  $\quiescence$-loops in states $s_0$, $s_4$, $s_5$, and $s_6$.
  Label $r{+}[a,b]$ should be interpreted as any value $r' \in [a+r, b+r]$
  and similarly $r{+}[a,b)$ and $r{+}(a,b]$, appropriately considering closed
  and open boundaries; 
``$\text{other}\_i$'' represents any other input not explicitly considered leaving the same state; and
  ``$\text{any}\_i$'' and ``$\text{any}\_o$'' represent any possible
  input and output (including $\quiescence$), respectively.
  In any case $\NoInp$ and $\NoOutp$ are not considered since they are
  not part of the alphabet of the LTS.  Also, we note that any
  possible sequence of inputs becomes enabled in the last states (omitted in the picture).

\denseparagraph{Robust cleanness of reference implementation.}
In the following, the aim is to show that $\Spec$ is robustly clean. 
By construction, each state in $\Spec$ equals the trace that
leads to that state. In other words, $\last(p)=\trace(p)$ for any 
$p \in \paths_*(\Spec)$ can be shown by induction. 
As a consequence, a path in $\Spec$ can be completely
identified by the trace it defines.
%
%
The following lemma states that $\Spec$ preserves 
all traces of the standard \qStd\ it is constructed from.
This can be proven by using that $\Std_\quiescence$ is robustly clean w.r.t.\ the (satisfiable) contract
$\Contract$ (see Def.~\ref{def:satisfiableContract}).

\begin{restatable}{lemma}{prespecPreservesTraces}
\label{lemma:prespec-preserves-traces}
Let $\Spec$ be constructed
  from contract $\Contract = \langle \Inputs, \Outputs, \Std, d_\Inputs, d_\Outputs \rangle$.
  Then
  $\traces_\omega(\Std_\quiescence)\subseteq\traces_\omega(\Spec)$.
  \end{restatable}

The following theorem states that the reference implementation
$\Spec$ is robustly clean w.r.t. the contract it was constructed from.

\begin{restatable}{theorem}{SCleanImpliesRClean}
\label{thm:SCleanImpliesRClean}
Let \Spec\ be constructed from \Contract. Then \Spec\ is robustly clean w.r.t. \Contract.
\end{restatable}

Furthermore, it is not difficult to show that $\Spec$ is indeed the largest implementation within the contract it was constructed from.

\begin{restatable}{theorem}{RefIsLargestImplementation}
\label{thm:testing:RefIsLargestImplementation}
Let \Spec\ be constructed from contract \Contract. Then \Spec\ is the largest implementation within $\Contract$.
\end{restatable}

\section{Model-Based Doping Tests}
\label{sec:testGeneration}

Following the conceptual ideas behind \ioco, we need to construct a specification that is compatible with our notion of robust
cleanness in such a way that a test suite can be derived. Intuitively, such a specification must be able to foresee every behaviour of the system that is allowed by the contract. We will take the reference implementation from the previous section as this specification. Indeed 
 we claim that $\Spec$ is constructed in 
such a way that whenever an IUT $\SUT$ is robustly clean, $\SUT\ioco\Spec$ holds. The latter translates to
\begin{equation}
  \forall \sigma \in \traces_*(\Spec_\quiescence):
  \out(\SUT_\quiescence \after \sigma) \subseteq \out(\Spec_\quiescence \after \sigma).
  \label{eq:doping:test}
\end{equation}

\vspace{-3mm}

\begin{restatable}{theorem}{cleanImpliesIoco}
\label{thm:testing:cleanImpliesIoco}
  Let $\Contract$ be a contract with standard $\Std$, let IOTS $\SUT$ be robustly clean w.r.t. $\Contract$ and with $\Std_\quiescence \subseteq \SUT_\quiescence$.
If $\Spec$ is constructed from $\Contract$, 
  then $\SUT \ioco \Spec$.
\end{restatable}
\noindent The key observations to prove this theorem are: 
\begin{inparaenum}[(i)]
\item the reference implementation is the largest implementation within the contract, i.e. if the IUT is robustly clean, then all its traces are covered by \Spec, and
\item
  by construction of $\Spec$ and satisfiability of \Contract, the suspension traces of $\Spec$ are exactly its finite traces.  
\end{inparaenum}

\begin{algorithm}[t]
    \caption{Doping Test (\DT)}
\algsetup{linenodelimiter=}
\algsetup{linenosize=\tiny}
\label{algo:dynamictest}
\begin{algorithmic}[1]
\REQUIRE history $h \in (\Inputs \cup \Outputs \cup \{ \quiescence \})^*$
\ENSURE \pass\ or \fail
\STATE c $\leftarrow$ $\Omega_\text{case}(h)$ \quad \COMMENT{Pick from one of three cases}
\IF{c = 1}
\RETURN \pass \quad \COMMENT{Finish test generation}
\ELSIF{c = 2 \AND no output from \SUT\ is available}
\STATE $i \leftarrow \Omega_\Inputs(h)$ \quad \COMMENT{Pick next input}
\STATE $i \twoheadrightarrow \SUT$ \quad \COMMENT{Forward input to SUT}
\RETURN $\DT(h \cdot i)$ \quad \COMMENT{Continue with next step}
\ELSIF{c = 3 \OR output from \SUT\ is available}
\STATE $o \twoheadleftarrow \SUT$ \quad \COMMENT{Receive output from SUT}
\IF{$o \in \clean(h)$}
\RETURN $\DT(h \cdot o)$ \  \COMMENT{If o is foreseen by oracle continue with next step}
\ELSE
\RETURN \fail	\quad \COMMENT{Otherwise, report test failure}
\ENDIF
\ENDIF
\end{algorithmic}
\end{algorithm}

\paragraph{Test Algorithm.}An important element of the model-based testing theory is a non-deterministic algorithm to generate test cases. A set of test cases is called a \emph{test suite}. 
It is shown elsewhere~\cite{DBLP:conf/fortest/Tretmans08}, that there is an algorithm that can produce a (possibly infinitely large) test suite $T$, for which
a system $\SUT$ passes $T$ if $\SUT$ is correct w.r.t. \ioco\ and, conversely, $\SUT$ is correct w.r.t. \ioco\ if $\SUT$ passes $T$. The former property is called \emph{soundness} and the latter is called \emph{exhaustiveness}. Algorithm~\ref{algo:dynamictest} shows a tail-recursive algorithm to test for robust cleanness. 
This \DT{} algorithm takes as an argument the history $h$ of the test currently running. Every doping test is inititalized by $\DT(\epsilon)$. 
Several runs of the algorithm constitute a test suite.
Each test can either \pass\ or \fail, which is reported by the output of the algorithm. In each call \DT\ picks one of three choices: 
\begin{inparaenum}[(i)]
\item it either terminates the test by returning \pass\ (line 3),
\item if there is no pending output that has to be read from the system under test, the algorithm may pick a new input and pass it to the system (lines 5-6), or
\item  \DT\ reads and checks the next output (or quiescence) that the system produces (lines 9-10).  Quiescence can be recognized by using a timeout mechanism that returns $\quiescence$ if no output has been received in a given amount of time.
\end{inparaenum}
In the original algorithm, the case and the next input are determined non-deterministically. Our algorithm is parameterized by $\Omega_\text{case}$ and $\Omega_\Inputs$, which can be instantiated by either non-determinism or some optimized test-case selection. Until further notice we assume non-deterministic selection. 
An output or quiescence that has been produced by the IUT is checked by means of an oracle $\clean$ (line 10). The oracle reflects the reference implementation \Spec, that is used as the specification for the \ioco\ relation and is defined in equation~(\ref{eq:algo:clean}).
\begin{align}
  \clean(h) \coloneqq \{
  & o \in \qOutputs \ | \label{eq:algo:clean}\\
  & \forall \sigma_i \in \mapInp{\traces_{\omega}(\Std_\quiescence)}:
  (\forall j \leq |h|{+}1: d_\Inputs(\mapInp{\sigma_i[..j]}, \mapInp{(h \cdot o)[..j]}) \leq \inpbound\})\notag\\
  & \qquad\!\! \Rightarrow \exists \sigma \in \traces_{\omega}(\Std_\quiescence): 
  \mapInp{\sigma} = \mapInp{\sigma_i} \land d_\qOutputs(o, \mapOut{\sigma[|h|+1]}) \leq \outpbound \}\notag
\end{align}
%
Given a finite execution, $\clean$ returns the set of acceptable outputs
(after such an execution) which corresponds exactly to the set of outputs
in $\Spec$ (after such an execution).  Thus $\clean(h)$ is precisely the
set of outputs that satisfies the premise in the definition of $\Spec$
after the trace $h$, as stipulated in Def.~\ref{def:testing:approxSpec}.

We refer to $\clean$ as an oracle, because it cannot be computed in general due to the infinite traces of $\Std_\quiescence$ in the definition. However, we get the following theorem stating that the algorithm is sound and exhaustive with respect to \ioco\ (and we present a computable algorithm in the next section). The theorem follows from the soundness and exhaustiveness of the original test generation algorithm for model-based testing and Def.~\ref{def:testing:approxSpec}.

\begin{restatable}{theorem}{iocoDT}
\label{thm:ioco-DT}
Let $\Contract$ be a contract with standard $\Std$. Let $\SUT$ be an implementation with $\Std_\quiescence \subseteq \SUT_\quiescence$
and let $\Spec$ be the largest implementation within $\Contract$.
Then, $\SUT \ioco \Spec$ if and only if for every test execution $t = \DT(\epsilon)$ it holds that $\SUT \passes t$.
\end{restatable}

\noindent Together with Theorem~\ref{thm:testing:cleanImpliesIoco} and satisfiability of \Contract, we derive the following
corollary.
\begin{corollary}
\label{cor:passing}
  Let $\Contract$ be a contract with standard $\Std$. Let $\SUT$ be an implementation with $\Std_\quiescence \subseteq \SUT_\quiescence$.
If $\SUT$ is
  robustly clean, then for every test execution $t = \DT(\epsilon)$ it holds that $\SUT \passes t$.
\end{corollary}
It is worth noting that in Corollary~\ref{cor:passing} we do not get that $\SUT$ is robustly clean if $\SUT$ always passes \DT. This is due the intricacies of genuine hyperproperties. By testing, we will never be able to verify the first condition of Def.~\ref{def:ed-clean:LTS}, because this needs a simultaneous view on all possible execution traces of \SUT. During testing, however, we always can observe only one trace.
 
\denseparagraph{Finite Doping Tests.}
As mentioned before, the execution of \DT\ is not possible, because the oracle $\clean$ is not computable. There is, however, a computable version $\clean_b$ of $\clean$ for executions up to some test length $b$ for bounded and discretised $\Inputs$ and $\Outputs$. Even for infinite executions, $b$ can be seen as a limit of interest and testing is still sound.
$\clean_b$ is shown in eq. (\ref{eq:algo:clean-approx}). The only variation w.r.t.\ $\clean$ lies in the use of the set
$\traces_{b}(\Std_\quiescence)$, instead of
$\traces_\omega(\Std_\quiescence)$, so as to return all traces of
$\Std_\quiescence$ whose length is exactly $b$.  Since
$\Std_\quiescence$ is finite, function $\clean_b$ can be implemented.
\begin{align}
  \clean_b(h) \coloneqq \{
  & o \in \qOutputs \ | \label{eq:algo:clean-approx}\\
  & \forall \sigma_i \in \mapInp{\traces_{b}(\Std_\quiescence)}: 
  (\forall j \leq |h|{+}1: d_\Inputs(\mapInp{\sigma_i[..j]}, \mapInp{(h \cdot o)[..j]}) \leq \inpbound)\notag\\
  & \qquad\!\! \Rightarrow \exists \sigma \in \traces_{b}(\Std_\quiescence):\ 
  \mapInp{\sigma} = \mapInp{\sigma_i} \land d_\qOutputs(o, \mapOut{\sigma[|h|{+}1]}) \leq \outpbound \}\notag
\end{align}
%
%
Now we get a new algorithm $\DT_b$ by replacing $\clean$ by $\clean_b$ in \DT\ and by forcing case 1 when and only when $|h| = b$. We get a similar soundness theorem for $\DT_b$ as in Corollary~\ref{cor:passing}.

\begin{restatable}{theorem}{boundedPassing}
\label{thm:boundedPassing}
 Let $\Contract$ be a contract with standard $\Std$. Let $\SUT$ be an implementation with $\Std_\quiescence \subseteq \SUT_\quiescence$.
 If $\SUT$ is
  robustly clean, then for every boundary $b$ and every test execution $t = \DT_b(\epsilon)$ it holds that $\SUT \passes t$.
\end{restatable}

Since $\SUT \passes \DT_b(\epsilon)$ implies  $\SUT \passes \DT_a(\epsilon)$ for any $a \leq b$, we have in summary arrived at an on-the-fly  algorithm $\DT_b$ that for sufficiently large $b$ (corresponding to the length of the test) will be able to conduct a ``convicting'' 
doping test for any IUT $\SUT$ that is not robustly clean w.r.t.\ a given contract $\Contract$.
The bounded-depth algorithm effectively circumvents the fact that, except for $\Std$ and $\Std_\quiescence$, all other objects we need to deal with are countably or  uncountably infinite and that the property we check is a hyperproperty.

We implemented a prototype of a testing framework using the bounded-depth algorithm.
The specification of distances, value domains and test case selection are parameters of the algorithm that can be set specific for a concrete test scenario.
This flexibility enables us to use the framework in a two-step approach for cyber-physical systems not equipped with a digital interface to forward the inputs to: first, the tool can generate test inputs, that are executed by a human or a robot on the CPS under test. The actual inputs (possibly deviating from the generated inputs) and outputs from the system are recorded so that in the second step our tool determines if the (actual) test is passed or failed. 

\section{Evaluation}
\label{sec:evaluation}
The normed emission test NEDC (New European Driving Cycle) (see Fig. \ref{fig:NEDC}) is the legally binding framework in Europe~\cite{nedc} (at the time the scandal surfaced). It is to be carried out on a chassis dynamometer and all relevant parameters are fixed by the norm,  including for instance the outside temperature at which it is run. 

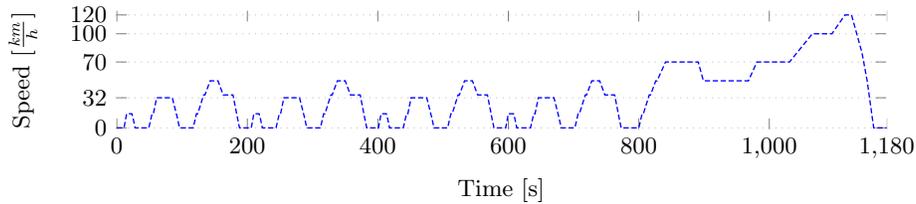
\begin{figure*}[tb]
\begin{center}
\hspace{1cm}    
\pgfplotsset{
    axis line style={white}
  }
\begin{tikzpicture}[trim axis left, scale=1]
\begin{axis}[
    width=.97\textwidth,
    height=0.26\textwidth,
    xlabel={Time [s]},
    ylabel={Speed [$\frac{\mathit{km}}{h}$]},
    xmin=-2, xmax=1180,
    ymin=-2, ymax=125,
    xtick={0,200, 400, 600, 800, 1000, 1180},
    ytick={0, 32, 70, 100, 120},
    legend pos=north west,
    ymajorgrids=true,
    grid style=dotted,draw=gray!10,
]
 
\addplot[
    color=blue,
    line width=.5pt,
    dash pattern=on 2pt off 1pt,
    ]
    coordinates {
    (0,0)(6,0)(11,0)(15,15)(23,15)(25,10)(28,0)(44,0)(49,0)(54,15)(56,15)(61,32)(85,32)(93,10)(96,0)(112,0)(117,0)(122,15)(124,15)(133,35)(135,35)(143,50)(155,50)(163,35)(176,35)(178,35)(185,10)(188,0)(195,0)(201,0)(206,0)(210,15)(218,15)(220,10)(223,0)(239,0)(244,0)(249,15)(251,15)(256,32)(280,32)(288,10)(291,0)(307,0)(312,0)(317,15)(319,15)(328,35)(330,35)(338,50)(350,50)(358,35)(371,35)(373,35)(380,10)(383,0)(390,0)(396,0)(401,0)(405,15)(413,15)(415,10)(418,0)(434,0)(439,0)(444,15)(446,15)(451,32)(475,32)(483,10)(486,0)(502,0)(507,0)(512,15)(514,15)(523,35)(525,35)(533,50)(545,50)(553,35)(566,35)(568,35)(575,10)(578,0)(585,0)(591,0)(596,0)(600,15)(608,15)(610,10)(613,0)(629,0)(634,0)(639,15)(641,15)(646,32)(670,32)(678,10)(681,0)(697,0)(702,0)(707,15)(709,15)(718,35)(720,35)(728,50)(740,50)(748,35)(761,35)(763,35)(770,10)(773,0)(780,0)(800,0)(805,15)(807,15)(816,35)(818,35)(826,50)(828,50)(841,70)(891,70)(895,60)(899,50)(968,50)(981,70)(1031,70)(1066,100)(1096,100)(1116,120)(1126,120)(1142,80)(1150,50)(1160,0)(1180,0)
    };

\end{axis}
\end{tikzpicture}
\caption{NEDC speed profile.}
\label{fig:NEDC}
\end{center}
\vspace{-1cm}
\end{figure*}


For a given car model, the normed test induces a standard LTS $\Std$ as follows. The input dimensions of $\Std$ are spanned by the sensors the car model is equipped with (including e.g. temperature of the exhaust, outside temperature, vertical and lateral acceleration, throttle position, time after engine start, engine rpm, possibly height above ground level etc.) which are accessible via the standardized OBD-2 interface~\cite{OBD2}. The output is the amount of \NOx\ per kilometre that has been extruded since engine start. 
Inputs are sampled  at equidistant times (once per second). The standard LTS $\Std$ is obtained from the trace representing the observations of running NEDC on the chassis dynamometer, say $\sigma_S \coloneqq i_1 \, \cdots \, i_{1180} \, o_S \, \quiescence  \, \quiescence \,  \quiescence \, \cdots$ with inputs $i_1, \cdots i_{1180}$  given by the NEDC over its 20 minutes (1180 seconds)
duration, and $o_S$ is the amount of \NOx\ gases accumulated during the test procedure. This  $\sigma_S$ is the only standard trace of our experiments. 
The trace ends with an infinite suffix $\quiescence^\omega$ of quiescence steps.  

The input space, $\Inputs$ is a  vector space spanned by  all possible input parameter dimensions. For $\vec{a}\in\Inputs$ we distinguish the speed dimension as $v(\vec{a})\in\Reals$  (measured in \textit{km/h}).  We can use past-forgetful distances with $d_\Inputs(\vec{a},\vec{b}) \coloneqq |v(\vec{a})-v(\vec{b})|$ if $\vec{a},\vec{b} \in \Inputs$, $d_\Inputs(\NoInp,\NoInp) = 0$ and $d_\Inputs(a,b) = \infty$ otherwise.
The speed is the decisive quantity defined to vary along the NEDC (cf.~Fig.~\ref{fig:NEDC}). Hence $d_\Inputs(\vec{a},\vec{b})=0$ if $v(\vec{a})=v(\vec{b})$ regardless of the values of other parameters. 
 We also take $\Outputs = \Reals$ for the average amount of \NOx\ gases per kilometre since engine start (in \textit{mg/km}). 
 We define $d_\Outputs(a,b)=|a-b|$ if $a,b\in\Outputs$, 
 and $d_\Outputs(a,b) = \infty$ otherwise. 

\denseparagraph{Doping tests in practice.}
 For the purpose of practically exercising doping tests, we picked a Renault 1.5 dci (110\textit{hp}) (Diesel) engine. 
This engine runs, among others, inside a Nissan NV200 Evalia which is classified as a Euro 6 car. The test cycle used in the original type approval of the car was \NEDC\ (which corresponds to Euro 6b). Emissions are cleaned using \emph{exhaust gas recirculation} (EGR). The technical core of EGR is a valve between the exhaust and intake pipe, controlled by a software. EGR is known to possibly cause performance losses, especially at higher speed. Car manufacturers might be tempted to optimize EGR usage  for engine performance unless facing a known test cycle such as the \NEDC.

We fixed a contract with $\inpbound = 15$ \textit{km/h}, $\outpbound = 180$ \textit{mg/km}. 
We report here on two of the tests we executed apart from the NEDC reference: \emph{(i)} \PowerNEDC\ is a variation of the \NEDC, where acceleration is increased from $0.94\frac{m}{s^2}$ to $1.5\frac{m}{s^2}$ in phase 6 of the \NEDC\ elementary urban cycle (i.e. after $56s, 251s, 446s$ and $641s$) and \emph{(ii)} \SineNEDC\ defines the speed at time $t$ to be the speed of the \NEDC\ at time $t$ plus $5 \cdot sin(0.5t)$ (but capped at $0$). Both can be generated by $\DT_{1181}(\epsilon)$ for specific deterministic $\Omega_\text{case}$ and $\Omega_\Inputs$. For instance, \SineNEDC\ is  given  below.  Fig. \ref{fig:SineNEDC} shows the initial 200s of \SineNEDC\ (red, dotted).
\vspace{-1mm}
\begin{align*}
\Omega_\text{case}(h) = \begin{cases}
2 & \text{, if } |h| \leq 1179 \\
3 & \text {, if } |h| = 1180
\end{cases} \quad
\Omega_\Inputs(h) = \max
\left\{
\begin{array}{ll}
0, \\
\text{\NEDC}(|h|) + 5 \cdot \sin(0.5|h|))
\end{array}
\right\}
\end{align*}

\vspace{-0mm}

\begin{figure}[t]
\begin{center}
\hspace{1cm}    
\pgfplotsset{
    axis line style={white}
  }
\begin{tikzpicture}[trim axis left, scale=1]
\begin{axis}[
    width=.95\textwidth,
    height=0.3\textwidth,
    xlabel={Time [s]},
    ylabel={Speed [$\frac{\mathit{km}}{h}$]},
    xmin=-2, xmax=200,
    ymin=-2, ymax=60,
    xtick={0, 100,200},
    ytick={0,15, 32, 50},
    legend pos=north west,
    ymajorgrids=true,
    grid style=dotted,draw=gray!10,
]
 
\addplot[
    color=blue,
    line width=.5pt,
    dash pattern=on 2pt off 1pt,
    ]
    coordinates {
    
    (0.0,-0.1440901)(1.0,-0.1706017)(2.0,-0.1706017)(3.0,-0.1440901)(4.0,0.01497977)(5.0,-0.06455518)(6.0,-0.09106683)(7.0,-0.1175785)(8.0,-0.1175785)(9.0,-0.09106685)(10.0,-0.09106686)(11.0,-0.1175785)(12.0,0.8103291)(13.0,9.983359)(14.0,10.06289)(15.0,12.26336)(16.0,13.82755)(17.0,13.66848)(18.0,13.40336)(19.0,13.82755)(20.0,15.10011)(21.0,15.78941)(22.0,15.7629)(23.0,13.00569)(24.0,10.69917)(25.0,10.64615)(26.0,10.64615)(27.0,6.510334)(28.0,-0.1440896)(29.0,-0.1440896)(30.0,-0.2501362)(31.0,-0.1440897)(32.0,-0.1706013)(33.0,-0.1706013)(34.0,-0.197113)(35.0,-0.0910664)(36.0,-0.1706013)(37.0,-0.03804312)(38.0,-0.09106642)(39.0,-0.1440897)(40.0,-0.1175781)(41.0,-0.1706014)(42.0,-0.1440897)(43.0,-0.06455481)(44.0,-0.09106646)(45.0,-0.09106647)(46.0,-0.09106648)(47.0,-0.03804319)(48.0,-0.1175781)(49.0,-0.1175781)(50.0,3.302424)(51.0,8.631266)(52.0,8.816847)(53.0,13.11173)(54.0,15.20615)(55.0,16.63778)(56.0,13.80104)(57.0,13.37685)(58.0,18.36104)(59.0,25.86383)(60.0,28.38244)(61.0,31.11314)(62.0,32.35919)(63.0,33.20757)(64.0,34.37408)(65.0,35.38151)(66.0,35.72617)(67.0,35.83221)(68.0,35.38152)(69.0,33.60523)(70.0,34.6657)(71.0,30.76849)(72.0,30.29128)(73.0,30.29128)(74.0,30.60942)(75.0,30.87453)(76.0,31.00709)(77.0,30.84802)(78.0,30.82152)(79.0,30.58291)(80.0,30.18524)(81.0,30.10571)(82.0,30.45036)(83.0,30.52989)(84.0,30.50338)(85.0,30.39734)(86.0,30.21176)(87.0,25.28059)(88.0,20.53501)(89.0,22.28477)(90.0,18.17547)(91.0,15.81593)(92.0,14.33128)(93.0,14.70244)(94.0,11.68011)(95.0,5.582437)(96.0,0.147549)(97.0,-0.1440791)(98.0,-0.1440791)(99.0,-0.1971024)(100.0,-0.1705908)(101.0,-0.09105587)(102.0,-0.1175675)(103.0,-0.1440792)(104.0,-0.1705908)(105.0,-0.1440792)(106.0,-0.03803262)(107.0,-0.1705909)(108.0,-0.1705909)(109.0,-0.1971025)(110.0,-0.1175676)(111.0,-0.1705909)(112.0,-0.2236142)(113.0,-0.1440793)(114.0,-0.1971026)(115.0,-0.1971026)(116.0,-0.1705909)(117.0,-0.1175677)(118.0,2.135922)(119.0,10.61965)(120.0,11.52105)(121.0,13.77454)(122.0,15.60384)(123.0,17.14152)(124.0,15.33872)(125.0,17.48617)(126.0,21.94012)(127.0,24.03454)(128.0,25.3071)(129.0,26.50013)(130.0,27.56059)(131.0,28.80664)(132.0,30.37082)(133.0,32.78338)(134.0,34.6392)(135.0,33.20757)(136.0,33.02198)(137.0,39.25222)(138.0,40.76339)(139.0,42.43362)(140.0,43.99781)(141.0,45.85363)(142.0,47.78898)(143.0,49.53875)(144.0,50.49316)(145.0,49.98944)(146.0,50.49316)(147.0,51.23549)(148.0,51.60665)(149.0,52.00432)(150.0,52.32246)(151.0,52.13688)(152.0,52.42852)(153.0,52.48154)(154.0,51.58014)(155.0,51.18247)(156.0,48.58433)(157.0,47.68294)(158.0,42.38061)(159.0,42.00944)(160.0,39.54387)(161.0,36.91921)(162.0,35.67316)(163.0,34.45363)(164.0,35.85875)(165.0,36.17689)(166.0,35.19596)(167.0,36.17689)(168.0,36.52154)(169.0,36.15038)(170.0,35.40805)(171.0,35.03689)(172.0,34.71875)(173.0,34.45363)(174.0,34.4006)(175.0,34.26805)(176.0,34.162)(177.0,33.89688)(178.0,33.73781)(179.0,29.65502)(180.0,18.70571)(181.0,22.41734)(182.0,17.43315)(183.0,18.38757)(184.0,15.4978)(185.0,15.07361)(186.0,12.13082)(187.0,4.442443)(188.0,-0.1440718)(189.0,-0.09104847)(190.0,-0.1705834)(191.0,-0.1440718)(192.0,-0.0910485)(193.0,-0.03802521)(194.0,-0.1705835)(195.0,-0.06453688)(196.0,-0.1440718)(197.0,-0.1440718)(198.0,-0.1175602)(199.0,-0.1440718)(200.0,-0.1970952)(201.0,-0.1440719)(202.0,-0.1175602)(203.0,-0.1440719)(204.0,-0.1175602)(205.0,-0.06453695)(206.0,-0.1440719)(207.0,4.972676)(208.0,6.80198)(209.0,8.73733)(210.0,12.05129)(211.0,13.69501)(212.0,14.14571)(213.0,14.72896)(214.0,15.60385)(215.0,16.16059)(216.0,16.74385)(217.0,16.74385)(218.0,16.34618)(219.0,14.5699)(220.0,9.824312)
    };

   \addplot [color=red,style=densely dotted,line width=.6pt ]
	coordinates {
	  (0,0)(1,2)(2,4)(3,5)(4,5)(5,3)(6,1)(7,0)(8,0)(9,0)(10,0)(11,0)(12,2)(13,9)(14,15)(15,20)(16,20)(17,19)(18,17)(19,15)(20,12)(21,11)(22,10)(23,11)(24,10)(25,10)(26,9)(27,7)(28,5)(29,5)(30,3)(31,1)(32,0)(33,0)(34,0)(35,0)(36,0)(37,0)(38,1)(39,3)(40,5)(41,5)(42,4)(43,2)(44,0)(45,0)(46,0)(47,0)(48,0)(49,0)(50,2)(51,8)(52,13)(53,17)(54,20)(55,19)(56,16)(57,17)(58,18)(59,20)(60,24)(61,28)(62,30)(63,32)(64,35)(65,36)(66,37)(67,36)(68,35)(69,32)(70,30)(71,28)(72,27)(73,27)(74,29)(75,31)(76,33)(77,36)(78,37)(79,37)(80,36)(81,34)(82,31)(83,29)(84,27)(85,27)(86,25)(87,24)(88,24)(89,24)(90,23)(91,21)(92,18)(93,13)(94,8)(95,2)(96,0)(97,0)(98,0)(99,0)(100,0)(101,1)(102,3)(103,5)(104,5)(105,4)(106,2)(107,0)(108,0)(109,0)(110,0)(111,0)(112,0)(113,0)(114,2)(115,4)(116,5)(117,5)(118,6)(119,7)(120,7)(121,8)(122,10)(123,10)(124,11)(125,16)(126,20)(127,25)(128,29)(129,31)(130,33)(131,33)(132,33)(133,33)(134,31)(135,30)(136,32)(137,36)(138,40)(139,44)(140,48)(141,51)(142,53)(143,53)(144,51)(145,49)(146,47)(147,45)(148,45)(149,46)(150,48)(151,51)(152,53)(153,54)(154,55)(155,54)(156,51)(157,46)(158,42)(159,38)(160,36)(161,34)(162,34)(163,34)(164,37)(165,39)(166,40)(167,40)(168,39)(169,37)(170,34)(171,32)(172,30)(173,30)(174,31)(175,33)(176,35)(177,38)(178,39)(179,36)(180,32)(181,27)(182,21)(183,15)(184,10)(185,5)(186,2)(187,0)(188,0)(189,1)(190,3)(191,5)(192,5)(193,4)(194,2)(195,0)(196,0)(197,0)(198,0)(199,0)(200,0)
	};
	
	\addplot [color=teal, style=solid,line width=.5pt]
	coordinates {
	(0.0,-0.1440901)(1.0,-0.06455515)(2.0,-0.1440901)(3.0,-0.1440901)(4.0,-0.09106681)(5.0,-0.1440901)(6.0,-0.09106683)(7.0,-0.09106684)(8.0,-0.1706018)(9.0,-0.06455521)(10.0,-0.1175785)(11.0,-0.09106687)(12.0,5.370332)(13.0,11.25592)(14.0,9.055451)(15.0,8.763824)(16.0,13.98662)(17.0,15.63034)(18.0,11.54755)(19.0,11.97173)(20.0,17.67174)(21.0,15.02057)(22.0,12.44894)(23.0,11.41499)(24.0,13.7215)(25.0,11.8922)(26.0,8.472195)(27.0,8.922892)(28.0,9.771264)(29.0,7.994984)(30.0,3.991725)(31.0,-0.1175798)(32.0,-0.1175798)(33.0,-0.1706031)(34.0,-0.1440915)(35.0,-0.1440915)(36.0,-0.1706031)(37.0,-0.1971148)(38.0,-0.09106822)(39.0,8.207077)(40.0,7.093588)(41.0,6.775448)(42.0,6.642891)(43.0,5.768007)(44.0,1.95033)(45.0,-0.09106726)(46.0,-0.06455563)(47.0,-0.1971139)(48.0,-0.1175789)(49.0,-0.1440906)(50.0,1.685213)(51.0,8.525218)(52.0,6.722426)(53.0,6.642891)(54.0,6.722426)(55.0,12.05127)(56.0,12.26336)(57.0,14.09266)(58.0,17.56569)(59.0,18.33453)(60.0,15.89546)(61.0,21.99314)(62.0,28.40895)(63.0,30.74198)(64.0,32.51826)(65.0,36.20338)(66.0,35.54058)(67.0,35.5671)(68.0,36.57454)(69.0,32.01453)(70.0,31.2457)(71.0,32.91594)(72.0,27.56058)(73.0,26.71221)(74.0,30.84802)(75.0,27.71965)(76.0,32.86291)(77.0,33.15454)(78.0,34.4271)(79.0,36.17686)(80.0,35.69965)(81.0,34.9043)(82.0,31.53732)(83.0,29.73454)(84.0,31.06012)(85.0,26.87128)(86.0,24.8829)(87.0,26.55314)(88.0,25.46616)(89.0,23.98151)(90.0,23.23918)(91.0,20.45546)(92.0,17.53918)(93.0,16.29313)(94.0,14.8615)(95.0,14.51685)(96.0,0.1475368)(97.0,-0.1175797)(98.0,-0.03804478)(99.0,-0.09106809)(100.0,-0.0380448)(101.0,-0.01153316)(102.0,7.067077)(103.0,6.006611)(104.0,6.669402)(105.0,5.635448)(106.0,-0.1175794)(107.0,-0.01153285)(108.0,-0.1440911)(109.0,-0.1440911)(110.0,-0.09106781)(111.0,-0.1175795)(112.0,-0.1706028)(113.0,0.0414904)(114.0,1.844282)(115.0,6.828472)(116.0,7.517775)(117.0,8.260101)(118.0,7.438239)(119.0,8.790333)(120.0,9.134985)(121.0,7.093587)(122.0,10.59313)(123.0,9.903822)(124.0,10.0894)(125.0,13.64196)(126.0,16.16057)(127.0,14.30476)(128.0,19.68662)(129.0,25.22755)(130.0,28.24988)(131.0,30.1057)(132.0,32.65081)(133.0,32.75686)(134.0,31.43128)(135.0,31.35175)(136.0,33.81733)(137.0,34.93082)(138.0,39.83547)(139.0,33.84384)(140.0,39.11966)(141.0,48.84943)(142.0,50.97036)(143.0,52.08385)(144.0,53.43594)(145.0,49.43269)(146.0,49.22059)(147.0,46.33082)(148.0,47.92152)(149.0,48.02757)(150.0,48.69036)(151.0,49.35314)(152.0,50.81129)(153.0,52.53454)(154.0,55.18571)(155.0,53.14431)(156.0,53.48896)(157.0,48.87594)(158.0,42.99035)(159.0,43.86523)(160.0,41.45267)(161.0,35.14291)(162.0,36.01779)(163.0,33.79081)(164.0,36.78662)(165.0,38.64244)(166.0,40.33918)(167.0,39.41127)(168.0,37.84709)(169.0,38.29778)(170.0,35.32848)(171.0,35.75267)(172.0,30.66243)(173.0,32.35918)(174.0,34.05592)(175.0,32.75685)(176.0,34.61267)(177.0,35.3815)(178.0,37.55546)(179.0,30.84801)(180.0,32.01452)(181.0,30.23825)(182.0,25.8108)(183.0,25.8108)(184.0,24.53824)(185.0,17.35359)(186.0,1.817761)(187.0,-0.1175897)(188.0,-0.170613)(189.0,-0.1175898)(190.0,6.377764)(191.0,7.040555)(192.0,7.040556)(193.0,5.052183)(194.0,2.772181)(195.0,0.7307842)(196.0,-0.1441002)(197.0,-0.09107696)(198.0,-0.09107696)(199.0,-0.01154204)(200.0,-0.1175886)(201.0,-0.09107699)(202.0,3.673577)(203.0,6.828463)(204.0,6.722417)(205.0,7.279162)
	};
\end{axis}
\end{tikzpicture}
\caption{Initial 200s of a \SineNEDC\ (red, dotted), its test drive  (green) and the  NEDC driven (blue, dashed).}
\label{fig:SineNEDC}
\vspace{-1cm}
\end{center}
\end{figure}
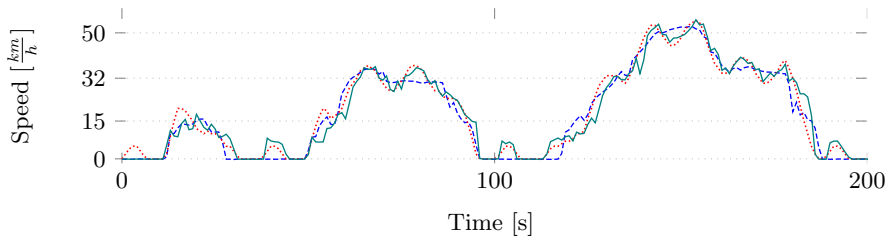

The car was fixed on a \emph{Maha LPS 2000} dynamometer and attached to an \emph{AVL M.O.V.E iS} portable emissions measurement system (PEMS, see Fig.~\ref{fig:evalia}) with speed data sampling at a rate of 20~\textit{Hz}, averaged to match the 1~\textit{Hz} rate of the NEDC. The human driver effectuated the \NEDC\ with a deviation of at most 9 \textit{km/h}  relative to  the reference 
 (notably, the result obtained for \NEDC\ are not consistent with 
  the car data sheet,
 likely caused by lacking calibration and absence of any further manufacturer-side optimisations).

\vspace{2mm}

 \hspace{-0.6cm}
\begin{minipage}{0.4\textwidth}

\includegraphics[width=1.1\columnwidth]{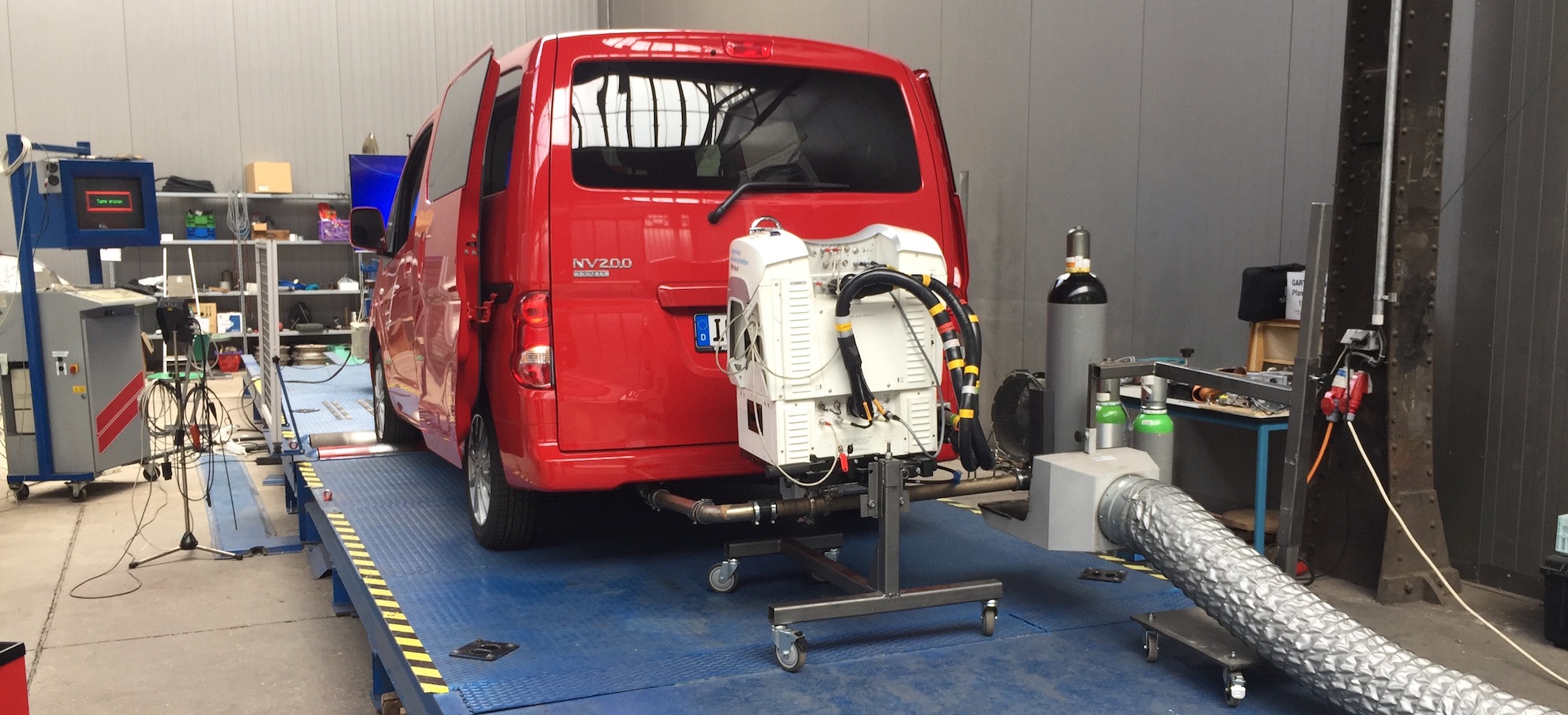}
\captionof{figure}{\label{fig:evalia}Nissan NV200 Evalia on a dynamometer}
\end{minipage}
\hspace{1cm}
\begin{minipage}{.6\textwidth}
\vspace{-0.003mm}
\renewcommand{\arraystretch}{1.123}
\begin{tabular}{@{}| l | r | r | r |@{}}
    \hline
    \textit{
    } & NEDC  & Power & Sine \\ \hline
    \textit{Distance} $\left[m\right]$ & 11,029 & 11,081 & 11,171 \\ \hline
    \textit{Avg. Speed} $\left[\frac{\textit{km}}{\textit{h}}\right]$ & 33 & 29 & 34 \\ \hline
     $\mathrm{CO}_2$ $\left[\frac{\textit{g}}{\textit{km}}\right]$ & 189 & 186  & 182 \\ \hline
     $\mathrm{NO}_x$ $\left[\frac{\textit{mg}}{\textit{km}}\right]$ & 180 & 204 & \textbf{584} \\
    \hline
    \end{tabular}

\captionof{table}{\label{tbl:htwresults}Dynamometer measurements \\(sample rate: 1\textit{Hz})}
\end{minipage}
\vspace{0mm}

\noindent The \PowerNEDC\ test drive differed by less than 15 \textit{km/h} and the \SineNEDC\ by less than 14 \textit{km/h} from the \NEDC\ test drive, so both inputs deviate by less than $\inpbound$. The green line in Fig. \ref{fig:SineNEDC} shows  \SineNEDC\ driven. The test outcomes are summarised in Table \ref{tbl:htwresults}. They show that the amount of CO$_2$ for the two tests is lower than for the \NEDC\ driven. The \NOx\ emissions of \PowerNEDC\ deviate by around 24 \textit{mg/km}, which is clearly below $\outpbound$. But the \SineNEDC\ produces about 3.24 times the amount of \NOx , that is 404 \textit{mg/km} more than what we measured for the \NEDC, which is a violation of the contract. This result can be verified with our algorithm a posteriori, namely by using $\Omega_\Inputs$ to replay the actually executed test inputs (which are different from the test inputs generated upfront due to human driving imprecisions) and by 
feeding the outputs recorded by the PEMS into the algorithm. As to be expected, this makes the recording of the \PowerNEDC\  return \pass\ and the recording of \SineNEDC\ return \fail.

Our algorithm is powerful enough to detect other kinds of defeat devices like those uncovered in investigations of the Volkswagen or the Audi case. Due to lack of space, we cannot present the concrete $\Omega_\text{case}$ and $\Omega_\Inputs$ for these examples.


\section{Discussion}
\label{sec:conclu}

\denseparagraph{Related Work.}
\label{sec:relwork}
The present work complements white-box approaches to software doping, like model-checking~\cite{DBLP:conf/esop/DArgenioBBFH17} or static code analysis~\cite{DBLP:conf/sp/ContagLPDLHS17} by a black-box testing approach, for which the specification is given implicitly by a contract, and usable for on-the-fly testing. 
Existing test frameworks like TGV \cite{DBLP:journals/sttt/JardJ05} or TorX \cite{DBLP:journals/sttt/VriesT00} provide support for the last step, however they fall short on scenarios where \begin{inparaenum}[(i)]\item the specification is not at hand and, among others, \item the test input is distorted in the testing process, e.g., by a human driving a car under test.
\end{inparaenum}

Our work is based on the definition of robust cleanness~\cite{DBLP:conf/esop/DArgenioBBFH17} which has conceptual similarities to continuity properties~\cite{DBLP:conf/popl/ChaudhuriGL10,DBLP:conf/issta/Hamlet02} of programs. However, continuity itself does not provide a reasonably good guarantee of cleanness. This is because physical outputs (e.g. the amount of \NOx\ gas in the exhaust) usually do change continuously. For instance, a doped car may alter its emission cleaning in a discrete way, but that induces a (rapid but) continuous change of \NOx\ gas concentrations. 
Established notions of stability and robustness~\cite{DBLP:conf/emsoft/TabuadaBCSM12,DBLP:conf/acsd/DoyenHLN10,DBLP:conf/rtss/MajumdarS09,572653} differ from robust cleanness in that the former assure the outputs (of a white-box system model) to stabilize despite transient input disturbances.  Robust cleanness does not consider perturbations but (intentionally) different inputs, and needs a hyperproperty formulation.

\denseparagraph{Concluding Remarks.}
This work lays the theoretical foundations for black-box testing approaches geared
towards uncovering doped software. As in the diesel emissions scandal -- where manufacturers were forced to pay excessive fines~\cite{VW-fine} and where executive managers are facing lawsuits or indeed went to prison~\cite{winterkorn-prison,stadler-prison} -- doped behaviour is typically strongly related to illegal behaviour.

As we have discussed, software doping analysis comes with several challenges. It can be performed   \begin{inparaenum}[(i)]\item only after production time on the final embedded or cyber-physical product, \item notoriously without support by the manufacturer, and \item the property belongs to the class of hyperproperties with alternating quantifiers. \item Non-determinism and imprecision caused by a human in-the-loop complicate doping analysis of CPS even further.\end{inparaenum}

Conceptually central to the approach is a contract that is assumed to be explicitly offered by the manufacturer. The contract itself is defined by very few parameters making it easy to form an opinion about a concrete contract. And even if a manufacturer is not willing to provide such contractual guarantees, instead a contract with very generous parameters can provide convicing evidence of doping if a test uncovers the contract violation. We showed this in a real automotive example demonstrating how a legally binding reference behaviour and a contract altogether induce a finite state LTS enabling to harvest input-output conformance testing for doping tests.
We developed an algorithm that can be attached directly to a system under test or in a three-step process, first generating a valid test case, afterwards used to guide a human interacting with the system, possibly adding distortions, followed by an a-posteriori validation of the recorded trajectory.
For more effective test case selection~\cite{DBLP:conf/pts/FeijsGMT02,TestPurposes} we are  exploring different guiding techniques~\cite{Faneikos-STALIRO,DeshmukhJKM15,DBLP:conf/cav/AdimoolamDDKJ17} for cyber-physical systems.  

\clearpage

\ifIEEE
\bibliographystyle{IEEEtran}
\IEEEtriggeratref{35}
\else
\bibliographystyle{splncs04}
\fi

\bibliography{softwaredoping}

\begin{thebibliography}{10}
\providecommand{\url}[1]{\texttt{#1}}
\providecommand{\urlprefix}{URL }
\providecommand{\doi}[1]{https://doi.org/#1}

\bibitem{DBLP:conf/cav/AdimoolamDDKJ17}
Adimoolam, A.S., Dang, T., Donz{\'{e}}, A., Kapinski, J., Jin, X.:
  Classification and coverage-based falsification for embedded control systems.
  In: Majumdar, R., Kuncak, V. (eds.) Computer Aided Verification - 29th
  International Conference, {CAV} 2017, Heidelberg, Germany, July 24-28, 2017,
  Proceedings, Part {I}. Lecture Notes in Computer Science, vol. 10426, pp.
  483--503. Springer (2017). \doi{10.1007/978-3-319-63387-9\_24},
  \url{https://doi.org/10.1007/978-3-319-63387-9\_24}

\bibitem{Faneikos-STALIRO}
Annpureddy, Y., Liu, C., Fainekos, G.E., Sankaranarayanan, S.: S-taliro: {A}
  tool for temporal logic falsification for hybrid systems. In: Abdulla, P.A.,
  Leino, K.R.M. (eds.) Tools and Algorithms for the Construction and Analysis
  of Systems - 17th International Conference, {TACAS} 2011, Held as Part of the
  Joint European Conferences on Theory and Practice of Software, {ETAPS} 2011,
  Saarbr{\"{u}}cken, Germany, March 26-April 3, 2011. Proceedings. Lecture
  Notes in Computer Science, vol.~6605, pp. 254--257. Springer (2011).
  \doi{10.1007/978-3-642-19835-9\_21},
  \url{https://doi.org/10.1007/978-3-642-19835-9\_21}

\bibitem{BartheDFH16:isola}
Barthe, G., D'Argenio, P.R., Finkbeiner, B., Hermanns, H.: Facets of software
  doping. In: Margaria and Steffen  \cite{DBLP:conf/isola/2016-2}, pp.
  601--608. \doi{10.1007/978-3-319-47169-3\_46},
  \url{http://dx.doi.org/10.1007/978-3-319-47169-3\_46}

\bibitem{DBLP:conf/isola/Baum16}
Baum, K.: What the hack is wrong with software doping? In: Margaria and Steffen
   \cite{DBLP:conf/isola/2016-2}, pp. 633--647.
  \doi{10.1007/978-3-319-47169-3\_49},
  \url{https://doi.org/10.1007/978-3-319-47169-3\_49}

\bibitem{stadler-prison}
{BBC}: {A}udi chief {R}upert {S}tadler arrested in diesel emissions probe.
  {BBC}, \url{https://www.bbc.com/news/business-44517753} (2018),
  \url{https://www.bbc.com/news/business-44517753}, {O}nline; accessed:
  2019-01-28

\bibitem{DBLP:conf/popl/ChaudhuriGL10}
Chaudhuri, S., Gulwani, S., Lublinerman, R.: Continuity analysis of programs.
  In: Hermenegildo, M.V., Palsberg, J. (eds.) Proceedings of the 37th {ACM}
  {SIGPLAN-SIGACT} Symposium on Principles of Programming Languages, {POPL}
  2010, Madrid, Spain, January 17-23, 2010. pp. 57--70. {ACM} (2010).
  \doi{10.1145/1706299.1706308},
  \url{http://doi.acm.org/10.1145/1706299.1706308}

\bibitem{ClarksonFKMRS14:post}
Clarkson, M.R., Finkbeiner, B., Koleini, M., Micinski, K.K., Rabe, M.N.,
  S{\'{a}}nchez, C.: Temporal logics for hyperproperties. In: Abadi, M.,
  Kremer, S. (eds.) {POST} 2014. LNCS, vol.~8414, pp. 265--284. Springer
  (2014). \doi{10.1007/978-3-642-54792-8\_15},
  \url{http://dx.doi.org/10.1007/978-3-642-54792-8\_15}

\bibitem{ClarksonS08}
Clarkson, M.R., Schneider, F.B.: Hyperproperties. In: CSF'08. pp. 51--65
  (2008). \doi{10.1109/CSF.2008.7}, \url{http://dx.doi.org/10.1109/CSF.2008.7}

\bibitem{DBLP:conf/sp/ContagLPDLHS17}
Contag, M., Li, G., Pawlowski, A., Domke, F., Levchenko, K., Holz, T., Savage,
  S.: How they did it: An analysis of emission defeat devices in modern
  automobiles. In: 2017 {IEEE} Symposium on Security and Privacy, {SP} 2017,
  San Jose, CA, USA, May 22-26, 2017. pp. 231--250. {IEEE} Computer Society
  (2017). \doi{10.1109/SP.2017.66}, \url{https://doi.org/10.1109/SP.2017.66}

\bibitem{DBLP:conf/esop/DArgenioBBFH17}
D'Argenio, P.R., Barthe, G., Biewer, S., Finkbeiner, B., Hermanns, H.: Is your
  software on dope? - {F}ormal analysis of surreptitiously "enhanced" programs.
  In: Yang, H. (ed.) Programming Languages and Systems - 26th European
  Symposium on Programming, {ESOP} 2017, Proceedings. Lecture Notes in Computer
  Science, vol. 10201, pp. 83--110. Springer (2017).
  \doi{10.1007/978-3-662-54434-1\_4},
  \url{https://doi.org/10.1007/978-3-662-54434-1\_4}

\bibitem{TestPurposes}
{de Vries}, R.: Towards formal test purposes. In: Brinksma, H., Tretmans, G.,
  Brinksma, H. (eds.) Formal Approaches to Testing of Software 2001 (FATES'01).
  pp. 61--76. No. NS-01-4 in BRICS Notes Series, BRICS, University of Aarhus (8
  2001)

\bibitem{DeshmukhJKM15}
Deshmukh, J.V., Jin, X., Kapinski, J., Maler, O.: Stochastic local search for
  falsification of hybrid systems. In: Finkbeiner, B., Pu, G., Zhang, L. (eds.)
  Automated Technology for Verification and Analysis - 13th International
  Symposium, {ATVA} 2015, Shanghai, China, October 12-15, 2015, Proceedings.
  Lecture Notes in Computer Science, vol.~9364, pp. 500--517. Springer (2015).
  \doi{10.1007/978-3-319-24953-7\_35},
  \url{https://doi.org/10.1007/978-3-319-24953-7\_35}

\bibitem{DBLP:conf/acsd/DoyenHLN10}
Doyen, L., Henzinger, T.A., Legay, A., Nickovic, D.: Robustness of sequential
  circuits. In: Gomes, L., Khomenko, V., Fernandes, J.M. (eds.) 10th
  International Conference on Application of Concurrency to System Design,
  {ACSD} 2010, Braga, Portugal, 21-25 June 2010. pp. 77--84. {IEEE} Computer
  Society (2010). \doi{10.1109/ACSD.2010.26},
  \url{https://doi.org/10.1109/ACSD.2010.26}

\bibitem{winterkorn-prison}
Ewing, J.: {E}x-{V}olkswagen {C.E.O.} {C}harged {W}ith {F}raud {O}ver {D}iesel
  {E}missions. {New York Times},
  \url{https://www.nytimes.com/2018/05/03/business/volkswagen-ceo-diesel-fraud.html}
  (2018),
  \url{https://www.nytimes.com/2018/05/03/business/volkswagen-ceo-diesel-fraud.html},
  {O}nline; accessed: 2019-01-28

\bibitem{DBLP:conf/pts/FeijsGMT02}
Feijs, L.M.G., Goga, N., Mauw, S., Tretmans, J.: Test selection, trace distance
  and heuristics. In: Schieferdecker, I., K{\"{o}}nig, H., Wolisz, A. (eds.)
  Testing of Communicating Systems XIV, Applications to Internet Technologies
  and Services, Proceedings of the {IFIP} 14th International Conference on
  Testing Communicating Systems - TestCom 2002, Berlin, Germany, March 19-22,
  2002. {IFIP} Conference Proceedings, vol.~210, pp. 267--282. Kluwer (2002)

\bibitem{FinkbeinerRS15:cav}
Finkbeiner, B., Rabe, M.N., S{\'{a}}nchez, C.: Algorithms for model checking
  {HyperLTL} and {HyperCTL$^*$}. In: Kroening, D., Pasareanu, C.S. (eds.) {CAV}
  2015. LNCS, vol.~9206, pp. 30--48. Springer (2015).
  \doi{10.1007/978-3-319-21690-4\_3},
  \url{http://dx.doi.org/10.1007/978-3-319-21690-4\_3}

\bibitem{DBLP:conf/issta/Hamlet02}
Hamlet, D.: Continuity in sofware systems. In: Frankl, P.G. (ed.) Proceedings
  of the International Symposium on Software Testing and Analysis, {ISSTA}
  2002, Roma, Italy, July 22-24, 2002. pp. 196--200. {ACM} (2002).
  \doi{10.1145/566172.566203}, \url{https://doi.org/10.1145/566172.566203}

\bibitem{DBLP:journals/sttt/JardJ05}
Jard, C., J{\'{e}}ron, T.: {TGV:} theory, principles and algorithms. {STTT}
  \textbf{7}(4),  297--315 (2005). \doi{10.1007/s10009-004-0153-x},
  \url{https://doi.org/10.1007/s10009-004-0153-x}

\bibitem{DBLP:conf/rtss/MajumdarS09}
Majumdar, R., Saha, I.: Symbolic robustness analysis. In: Baker, T.P. (ed.)
  Proceedings of the 30th {IEEE} Real-Time Systems Symposium, {RTSS} 2009,
  Washington, DC, USA, 1-4 December 2009. pp. 355--363. {IEEE} Computer Society
  (2009). \doi{10.1109/RTSS.2009.17},
  \url{https://doi.org/10.1109/RTSS.2009.17}

\bibitem{DBLP:conf/isola/2016-2}
Margaria, T., Steffen, B. (eds.): Leveraging Applications of Formal Methods,
  Verification and Validation: Discussion, Dissemination, Applications - 7th
  International Symposium, ISoLA 2016, Part {II}, LNCS, vol.~9953 (2016).
  \doi{10.1007/978-3-319-47169-3},
  \url{http://dx.doi.org/10.1007/978-3-319-47169-3}

\bibitem{572653}
Pettersson, S., Lennartson, B.: Stability and robustness for hybrid systems.
  In: Proceedings of 35th IEEE Conference on Decision and Control. vol.~2, pp.
  1202--1207 vol.2 (Dec 1996). \doi{10.1109/CDC.1996.572653}

\bibitem{VW-fine}
Riley, C.: {V}olkswagen's diesel scandal costs hit \$30 billion. {CNN
  Business},
  \url{https://money.cnn.com/2017/09/29/investing/volkswagen-diesel-cost-30-billion/index.html}
  (2018),
  \url{https://money.cnn.com/2017/09/29/investing/volkswagen-diesel-cost-30-billion/index.html},
  {O}nline; accessed: 2019-01-28

\bibitem{DBLP:conf/emsoft/TabuadaBCSM12}
Tabuada, P., Balkan, A., Caliskan, S.Y., Shoukry, Y., Majumdar, R.:
  Input-output robustness for discrete systems. In: Jerraya, A., Carloni, L.P.,
  Maraninchi, F., Regehr, J. (eds.) Proceedings of the 12th International
  Conference on Embedded Software, {EMSOFT} 2012, part of the Eighth Embedded
  Systems Week, ESWeek 2012, Tampere, Finland, October 7-12, 2012. pp.
  217--226. {ACM} (2012). \doi{10.1145/2380356.2380396},
  \url{http://doi.acm.org/10.1145/2380356.2380396}

\bibitem{OBD2}
{The European Parliament and the Council of the European Union}: Directive
  98/69/ec of the european parliament and of the council. Official Journal of
  the European Communities  (1998),
  \url{http://eur-lex.europa.eu/LexUriServ/LexUriServ.do?uri=CELEX:31998L0069:EN:HTML}

\bibitem{DBLP:phd/basesearch/Tretmans92}
Tretmans, J.: A formal approach to conformance testing. Ph.D. thesis,
  University of Twente, Enschede, Netherlands (1992),
  \url{http://purl.utwente.nl/publications/58114}

\bibitem{DBLP:journals/cn/Tretmans96}
Tretmans, J.: Conformance testing with labelled transition systems:
  Implementation relations and test generation. Computer Networks and {ISDN}
  Systems  \textbf{29}(1),  49--79 (1996). \doi{10.1016/S0169-7552(96)00017-7},
  \url{https://doi.org/10.1016/S0169-7552(96)00017-7}

\bibitem{DBLP:conf/fortest/Tretmans08}
Tretmans, J.: Model based testing with labelled transition systems. In:
  Hierons, R.M., Bowen, J.P., Harman, M. (eds.) Formal Methods and Testing, An
  Outcome of the {FORTEST} Network, Revised Selected Papers. Lecture Notes in
  Computer Science, vol.~4949, pp. 1--38. Springer (2008).
  \doi{10.1007/978-3-540-78917-8\_1},
  \url{https://doi.org/10.1007/978-3-540-78917-8\_1}

\bibitem{nedc}
{United Nations}: {UN Vehicle Regulations - 1958 Agreement, Revision 2,
  Addendum 100, Regulation No. 101, Revision 3 ---
  E/ECE/324/Rev.2/Add.100/Rev.3} (2013),
  \url{http://www.unece.org/trans/main/wp29/wp29regs101-120.html}

\bibitem{DBLP:journals/sttt/VriesT00}
de~Vries, R.G., Tretmans, J.: On-the-fly conformance testing using {SPIN}.
  {STTT}  \textbf{2}(4),  382--393 (2000). \doi{10.1007/s100090050044},
  \url{https://doi.org/10.1007/s100090050044}

\end{thebibliography}

\clearpage

\newcommand{\iproof}[1]{Proof #1:}

\ifaddappendix

\appendix
\noindent This appendix contains 
\iftoolprototype
additional information about our implementation and 
\fi
proofs of all results claimed in this paper (for the benefit of only the reviewing process).

\iftoolprototype
\section{Prototype Implementation}
\label{sec:apendix:implementation}

To demonstrate that our theory is implementable, we provide a prototype implementation of a testing framework written in Python. The core implementation contains several abstract classes leaving the choice of value domains, distance functions, the system under test and test case selection unspecified. We show examples on how to instantiate the framework by means of subclasses for concrete values and distances and demonstrate different kinds of test case selections using an easy to understand program.

A related technique to testing is \emph{monitoring}. A monitor can read the inputs and outputs of a system in order to detect incorrect behaviour of the system. In contrast to testing, the inputs are not provided by the test, but the system is monitored during normal operation. Monitors can be either online (evaluation is done while inputs are still received) or offline (observed behaviour is evaluated after the observation). A monitor can easily be extended to a test by controlling the environment providing the inputs to the system.

Technically, our test in Section \ref{sec:evaluation} is somewhere between monitoring and testing. Like for testing, we define a test by means of a test cycle (generated by using deterministic $\Omega_\text{case}$ and $\Omega_\Inputs$ for test case selection). However, we do not have full control over passing the inputs to the system under test, i.e. the car, because someone has to drive the test cycle manually. This is  done as as precise as possible while values of the actual (speed) inputs and the (\NOx) outputs are recorded. The result is a trace for which we  do offline monitoring (see Sec.~\ref{sec:doping_monitor}) to check a posteriori whether the test is passed or failed -- or trivially passed, meaning that the input threshold has been violated.

We extended our testing framework to support offline monitoring by specifying a system under test and a test case selection that is used to analyse the recorded trace. We instantiated the monitor to show that \PowerNEDC\ passes the test and \SineNEDC\ fails.

The repository is organised as follows.

\begin{itemize}
\item \texttt{doping\_test} contains all classes for the abstract testing framework
\item \texttt{doping\_monitor} is the extension of \texttt{doping\_test} to an abstract offline monitor
\item \texttt{examples} contains several examples where values have type \texttt{float} and distances are defined as $d(t_1, t_2) \coloneqq \abs{\text{last}(t_1) - \text{last}(t_2)}$. Most examples are easy-to-understand programs. We also present the experiments from Section~\ref{sec:evaluation} in \texttt{examples/nissan}.
\end{itemize}

\subsection{Abstract Testing Framework}
The abstract testing framework is defined in \texttt{doping\_test}. The main class is \texttt{DT}, which implements the bounded variant of the $\DT$ algorithm from Section \ref{sec:testGeneration}. It works with instances of \texttt{SystemUnderTest} and \texttt{TestCaseSelection}, which are abstract interfaces for a system under test $\SUT$ and the selection of $\Omega_\text{case}$ and $\Omega_\Inputs$. The outputs from $\SUT$ are checked by an instance of \texttt{AcceptanceChecker}. The default implementation of \texttt{AcceptanceChecker} determines if an output $o$ received from the system under test is clean w.r.t. to the test history $h$, i.e. it checks whether $o \in \clean_b(h)$.

Symbols are instances of \texttt{Input} or \texttt{Output}, which wrap a concrete value and allow for convenient distinction between inputs and outputs throughout the framework. The classes \texttt{Distance}, \texttt{ValueSet}, \texttt{Standard} and \texttt{Trace} define interfaces for distance functions, sets of values, a standard $\Std$ and traces. A concrete implementation for \texttt{ValueSet} is \texttt{EmptySet} for the empty set and \texttt{SimpleTrace} wrapping an array of symbols for \texttt{Trace}. We provide an implementation for random testing in class  \texttt{RandomTestCaseSelection}, which picks one of the three choices of $\DT$ randomly with different probabilities. For the selection of test inputs it tries to pick values that do not immediately let the test pass because it is too far away from a standard input. Beyond that, the choice of inputs is random with the distribution defined by the concrete subclass of \texttt{ValueSet} being used.

\subsection{Abstract Doping Monitor}
\label{sec:doping_monitor}
We extended the doping test framework by \texttt{doping\_monitor} for offline monitoring of traces. The central class is \texttt{RecordedTrace}, which can load recorded traces from disk and can provide a system under test $\SUT$ of class \texttt{MonitorUnderTest}, which simulates exactly the recorded behaviour when provided with the recorded inputs. The correct choice of inputs is guaranteed by the provided instance of class \texttt{MonitorTestCaseSelection} (defining $\Omega_\text{case}$ and $\Omega_\Inputs$ as explained in Sec.~\ref{sec:evaluation}), which makes $\DT$ pick exactly the inputs that were recorded or to wait for an output if an output has been recorded. For constructing a standard $\Std$ from (possibly several) recorded traces, an array of recorded traces can be passed to the constructor of \texttt{MonitoredStandard}. For using the framework it is necessary to provide an implementation to \texttt{RecordedTrace} for parsing the encoded symbols in a file.

\subsection{Example: Noisy Mirror}
The class \texttt{NoisyMirror} in \texttt{examples/numbers} implements a function that randomly either returns the input or twice as much. However, the program is doped. A shady programmer observed that the official authority testing the software uses inputs that have at most two decimals. This observation seduced him to add a check if more than two decimals are passed to the program in order to output values up to four times as much as the input. The official standard test commits positive integers in ascending order to the program (i.e. 1, 2, 3, ...). The resulting traces are provided by \texttt{NoisyStandard}. The sets of values used in our examples are closed ranges of floats, hence we extended \texttt{ValueSet} to \texttt{NumberRange}. We use the past-forgetful distance function $d(t_1, t_2) \coloneqq \abs{\text{last}(t_1) - \text{last}(t_2)}$ for both inputs and outputs and it is implemented in \texttt{LastComponentDistance}. We test a contract with $\inpbound = 0.2$ and $\outpbound = 0.5$.

The program in \texttt{randomTest.py} uses the random test case selection from \texttt{doping\_test} to test the NoisyMirror. Running the test several times shows that it sometimes passes and sometimes fails. In \texttt{predefinedTest.py} we demonstrate how to feed $\DT$ with inputs that have been hand-picked before the test. It uses \texttt{ManualTestCaseSelection} to generate one of four tests. We show two well defined tests: \emph{good-test1} passes, because inputs always have one decimal and \emph{good-test2} eventually has more than two decimals and mostly fails. The other two tests are badly designed as they both pass, because they are trivially satisfied. \emph{bad-test1} provides an input that is too far away from a standard and the test will pass although outputs are suspiciously large. In \emph{bad-test2}, the test waits for an output when it is supposed to provide an input, which leads to an input distance of infinity, so the test passes trivially, too. The third example is in \texttt{monitoring.py}, which shows how offline monitoring works. We construct the same standard as for the previous examples by using a \texttt{MonitoredStandard} with the traces recorded in \texttt{run-std1.txt} and \texttt{run-std2.txt}. The trace recorded in \texttt{run-fail.txt} will fail (because the output in line 16 has no standard output in $\outpbound$-distance)  and for \texttt{run-pass.txt} $\DT$ passes.


\subsection{Example: Doped Nissan}
The a posteriori processing of our examples with a Nissan from Sec. \ref{sec:evaluation} is located in \texttt{examples/nissan}. The approach is the same as for the NoisyMirror monitoring example. Each of the files \texttt{NEDC.txt}, \texttt{PowerNEDC.txt} and \texttt{SineNEDC.txt} contains 1180 inputs (the speed) and returns the \NOx-production in mg/km in line 1181. We create a standard with the single trace recorded in \texttt{NEDC.txt}. We have the same distance function as for the NoisyMirror example and $\inpbound = 15$ and $\outpbound = 88$, as described in Sec.~\ref{sec:evaluation}. The \PowerNEDC\ is tested in \texttt{checkPowerNEDC.py} and passes. \texttt{checkSineNEDC.py} checks the recorded \SineNEDC-drive and reports about the test failure as expected.

\subsection{How to run}
The tool can be used by executing the Python files from folder \texttt{tools}.

\begin{itemize}
\item \texttt{python examples/numbers/randomTest.py}
\item \texttt{python examples/numbers/predefinedTest.py}
\item \texttt{python examples/numbers/monitoring.py}
\item \texttt{python examples/nissan/checkPowerNEDC.py}
\item \texttt{python examples/nissan/checkSineNEDC.py}
\end{itemize}

The tool will either say \emph{Test passed!} and show the test trace or it says \emph{Test FAILED for Standard Trace} and additionally gives a standard trace with an input distance of at most $\inpbound$ up to the length of the test trace. However, there is no standard trace with the same inputs as the standard trace shown, which has an output of at most $\outpbound$.
Notice that the tool shows a warning when the test is trivially passed due to a test input deviating by more than $\inpbound$ from all standard traces.

\clearpage
\fi

\section{Proofs of Theorems and Lemmas}
\label{app:proofs}

We summarize the assumptions that we use during the paper and which we will need for the following proofs.
\begin{enumerate}
\item $d_\qOutputs$ is past-forgetful, i.e. $d_\Outputs(\sigma_1,\sigma_2)=d_\Outputs(\sigma'_1,\sigma'_2)$
whenever $\last(\sigma_1)=\last(\sigma'_1)$ and \label{assume:past-forgetful}
$\last(\sigma_2)=\last(\sigma'_2)$,
\item
For every $o$ it holds that $d_\qOutputs(o, \quiescence) = d_\qOutputs(\quiescence, o) = \infty$, \label{assume:quiescence-o-infinite}
\item $\Std$ is a finite LTS, \label{assume:Std-finite}
\item \Contract\ is satisfiable, i.e. \label{assume:contract-satisfiable}
\begin{enumerate}
\item every input $\sigma_i \in (\Inputs \cup \{ \NoInp \})^\omega$ is satisfiable, and \label{assume:inputs-satisfiable}
\item \qStd\ is robustly clean w.r.t. $\Contract_\Std$, \label{assume:qstd-clean}
\end{enumerate}
\item $\Omega_\text{case}$ and $\Omega_\Inputs$ are non-deterministic. \label{assume:omega_nondet}
\end{enumerate}

\subsection{\Spec\ is robustly clean and the largest implementation}
The first set of proofs aims to show that \Spec\ is reasonably defined. That is, \Spec\ constructed from \Contract\ is supposed to be the largest implementation that is robustly clean w.r.t. \Contract.

The construction of \Spec\ is such that we are always able to identify for each state the (unique) trace which can reach that state and, conversely, for each trace, we know the (unique) state that can be reached by that trace. This is shown in Lemma~\ref{lemma:testing:RefTraceIsState} and Corollary~\ref{cor:RefTrace}.

\begin{lemma} \label{lemma:testing:RefTraceIsState}
  For all $p \in \paths_*(\Spec)$,
  $\last(p)=\trace(p)$.
\end{lemma}

\begin{proof}
  We proceed by induction on the numbers of states in $p$.  If $p$ has
  only one state then $\last(p)=p=\epsilon=\trace(p)$, since
  $\epsilon$ is the initial state in $\Spec$.

  Suppose now $p=p'\,a\,s\in\paths_*(\Spec)$.  By induction,
  $\last(p')=\trace(p')$.  By Def.~\ref{def:testing:approxSpec},
  $\last(p')\xrightarrow{\;a\;}_{\Spec} s$ only if
  $s=\last(p')\cdot a$.  But
  $\last(p)=s=\last(p')\cdot a = \trace(p')\cdot a = \trace(p)$,
  which proves the lemma.
  \qed
\end{proof}

\begin{corollary}\label{cor:RefTrace}
  Let $p \in \paths_*(\Spec)$ and $\sigma=\trace(p)$.  Then
  $p$ is exactly the path
  $\epsilon\,\sigma[1]\,(\sigma[..1])\,\sigma[2]\,(\sigma[..2])
  \cdots
  (\sigma[..|\sigma|-1])\,\sigma[|\sigma|]\,(\sigma[..|\sigma|])$.
\end{corollary}

The minimum requirement for any implementation $\LTS$ with some standard $\Std$ is that the behaviour of \Std\ is fully contained in $\LTS$. Since we argue about behaviour in terms of traces, we need to consider the quiescences closures $\LTS_\quiescence$ and $\qStd$. Lemma~\ref{lemma:prespec-preserves-traces} (appearing in the main paper) shows that \Spec\ satisfies this property.

\prespecPreservesTraces*

\begin{proof}
Since $\qStd$ is finite (Assumption~\ref{assume:Std-finite}) we can construct a deterministic (image-finite) LTS (where states are finite traces of the original) $\qStd'$ with $\traces_\omega(\qStd) = \traces_\omega(\qStd')$. With $\Spec$ being deterministic, too,
we only need to
  prove $\traces_*(\qStd)\subseteq\traces_*(\Spec)$
  (see~\cite{vanGlabbeek2001:handbook} for a proof).

  Let $\sigma\in\traces_*(\mStdPlus)$.  We proceed by induction on
  $k=|\sigma|$.  If $k = 0$ then $\sigma = \epsilon$ and hence
  $\sigma\in\traces_*(\Spec)$.
  If $k>0$, we know that $\sigma[..k-1]\in\traces_*(\qStd)$ and from the
  inductive hypothesis that
  $\sigma[..k-1]\in\traces_*(\Spec)$.
  By Lemma~\ref{lemma:testing:RefTraceIsState}, there is a path
  $p\in\paths_*(\Spec)$ with
  $\last(p)=\sigma[..k-1]$ and $\trace(p) = \sigma[..k-1]$.
  To show that $\sigma\in\traces_*(\Spec)$, we need to show
  that there is a transition
  $\sigma[..k-1] \xrightarrow{\sigma[k]}_{\Spec} \sigma$
  in $\Spec$.
  By Def.~\ref{def:testing:approxSpec}, this holds if for any
  $\sigma_i \in \mapInp{\traces_\omega(\mStdPlus)}$ with
  $\forall j \leq k: d_\Inputs(\mapInp{\sigma[..j]}, \sigma_i[..j]) \leq \inpbound$,
  there is some $\sigma_S \in \traces_\omega(\mStdPlus)$ with
  $\mapInp{\sigma_S}=\sigma_i$ for which
  $d_\qOutputs(\mapOut{\sigma[k]}, \mapOut{\sigma_S[k]}) \leq \outpbound$.
  The existence of $\sigma_i$ implies the existence of some
  $\sigma_{io}\in\traces_\omega(\mStdPlus)$ with
  $\mapInp{\sigma_{io}}=\sigma_i$.
  Also, notice that $\sigma$ can be extended to an infinite trace
  $\sigma'\in\traces_\omega(\mStdPlus)$ such that
  $\sigma'[..k]=\sigma$.
  By assumption~\ref{assume:qstd-clean} we know that $\mStdPlus$ is robustly clean w.r.t. $\Contract_\Std$.
  Then, by 
  Def.~\ref{def:ed-clean:LTS}.2, there is some
  $\sigma''\in\traces_\omega(\mStdPlus)$ with
  $\mapInp{\sigma''} = \mapInp{\sigma_{io}} \ (= \sigma_i)$ and
  $d_\Outputs(\mapOut{\sigma'[k]}, \mapOut{\sigma''[k]}) \leq \outpbound$.
  Taking $\sigma_S=\sigma''$ concludes the proof.
%
\qed
\end{proof}

In order to prove robust cleanness of \Spec, we first show in the following lemma that $\Spec$ satisfies the second condition
of Def.~\ref{def:ed-clean:LTS}.

\begin{restatable}{lemma}{RefSatisfiesSecondRule}
\label{lemma:testing:RefSatisfiesRule2}
  Let \Spec\ be constructed from contract \Contract.  Then, for all
  $\sigma,\sigma'\in\traces_\omega(\Spec)$, if
  $\sigma\in\traces_\omega(\qStd)$, it holds that for all $k\geq 0$
  such that
  $d_\Inputs(\mapInp{\sigma[..j]}, \mapInp{\sigma'[..j]}) \leq \inpbound$
  for all $j\leq k$, there exists
  $\sigma''\in\traces_\omega(\Spec)$ such that
  $\mapInp{\sigma}=\mapInp{\sigma''}$ and
  $d_\qOutputs(\mapOut{\sigma'[k]}, \mapOut{\sigma''[k]}) \leq \outpbound$.
\end{restatable}
\begin{proof}
\begin{sloppypar}
  Let $\sigma,\sigma'\in\traces_\omega(\Spec)$ with
  $\sigma\in\traces_\omega(\qStd)$.
%
  By Corollary~\ref{cor:RefTrace}, there is a path
  $p =
  \epsilon\,\sigma'[1](\sigma'[..1])\ldots(\sigma'[..k-1])\sigma'[k](\sigma'[..k])
  \in \paths_*(\Spec)$.
  In particular
  $\sigma'[..k-1]\xrightarrow{\sigma'[k]}\sigma'[..k]$
  is a transition in $\Spec$.
  By Def.~\ref{def:testing:approxSpec}, we know that for all
  $\sigma_i\in\mapInp{\traces(\mStdPlus)}$ such that for all j $\leq k$
  $d_\Inputs(\mapInp{\sigma'[..k][..j]},\sigma_i[..j])\leq\inpbound$
  (which is equivalent to
  $d_\Inputs(\mapInp{\sigma'[..j]},\sigma_i[..j])\leq\inpbound$),
  there is some $\sigma_S\in\traces_\omega(\mStdPlus)$ with
  $\mapInp{\sigma_S}=\sigma_i$ and
  $d_\qOutputs(\mapOut{\sigma[k]}, \mapOut{\sigma_S[k]}) \leq \outpbound$ (*).

  Since $\sigma\in\traces_\omega(\mStdPlus)$ then
  $\mapInp{\sigma}\in\mapInp{\traces_\omega(\mStdPlus)}$.
  Suppose that $d_\Inputs(\mapInp{\sigma'[..j]},\mapInp{\sigma[..j]})\leq\inpbound$ for
  all $j\leq k$ (otherwise the lemma holds trivially).
  Then, by (*), there exist $\sigma_S\in\traces_\omega(\mStdPlus)$ with
  $\mapInp{\sigma_S}=\mapInp{\sigma}$ such that
  $d_\qOutputs(\mapOut{\sigma[k]}, \mapOut{\sigma_S[k]}) \leq \outpbound$.
  Since
  $\traces_\omega(\qStd)\subseteq\traces_\omega(\Spec)$
  by Lemma~\ref{lemma:prespec-preserves-traces}, the lemma follows.  
\qed
  \end{sloppypar}
\end{proof}

The following Lemma shows that \Spec\ is the largest implementation that satisfies Def.~\ref{def:ed-clean:LTS}.\ref{def:ed-clean:LTS:ii} (the second condition of robust cleanness).
\begin{lemma}
\label{lemma:spec-is-largest-clean2}
Let \Contract\ be a contract and \Spec\ be constructed from \Contract. Then, for all $\LTS$ satisfying Def.~\ref{def:ed-clean:LTS}.\ref{def:ed-clean:LTS:ii}, $\traces_\omega(\LTS_\quiescence) \subseteq \traces_\omega(\Spec)$.
\end{lemma}
\begin{proof}
For a proof by contradiction, suppose that there is some $\LTS$ satisfying Def.~\ref{def:ed-clean:LTS}.\ref{def:ed-clean:LTS:ii}, but which  has some trace $\sigma \in \traces_\omega(\LTS_\quiescence)$ that is not a trace of \Spec, i.e. $\sigma \not\in \traces_\omega(\Spec)$.
Let $\traces_*(\Spec)$ be the finite prefixes of \Spec. Since $\sigma \not\in \traces_\omega(\Spec)$, there must be some $k > 0$ for which $\sigma[..k-1] \in \traces_*(\Spec)$, but $\sigma[..k] \not\in \traces_*(\Spec)$. Hence, there is no transition $\sigma[..k-1] \xrightarrow{\sigma[k]}_\Spec \sigma[..k]$ in \Spec.
This can only be, because the premise of Def.~\ref{def:testing:approxSpec} is not satisfied, i.e. there is some $\sigma_i \in \mapInp{\traces_\omega(\qStd)}$, such that for all $j \leq k$, $d_\Inputs(\mapInp{\sigma[..j]}, \sigma_i[..j]) \leq \inpbound$ (*) and for all standard traces $\sigma_S \in \traces_\omega(\mStdPlus)$ with $\mapInp{\sigma_s} = \sigma_i$ it holds that $d_\qOutputs(\mapOut{\sigma[k]}, \mapOut{\sigma_S[k]}) > \outpbound$ (**).

For $\sigma \in \traces_\omega(\LTS_\quiescence)$ there must be a transition $\sigma[..k-1] \xrightarrow{\sigma[k]}_\LTS \sigma[..k]$ in $\LTS_\quiescence$.
Let $\sigma_{io} \in \traces_\omega(\qStd)$ such that $\mapInp{\sigma_{io}} = \sigma_i$. From Def.~\ref{def:ed-clean:LTS}.\ref{def:ed-clean:LTS:ii} we get for $\LTS, \sigma_{io}, \sigma$ and $k$ with (*) a trace $\sigma'' \in \traces_\omega(\LTS_\quiescence)$ with $\mapInp{\sigma''} = \mapInp{\sigma_{io}} = \sigma_i$ and $d_\qOutputs(\mapOut{\sigma[..k]}, \mapOut{\sigma''[..k]}) \leq \outpbound$.
Furthermore, we get from Def.~\ref{def:LTS:Std}, $\sigma_i \in \traces_\omega(\qStd)$, $\sigma'' \in \traces_\omega(\LTS_\quiescence)$ and $\mapInp{\sigma''} = \sigma_i$ that $\sigma'' \in \traces_\omega(\qStd)$. This is a contradiction to (**).
\qed
\end{proof}

The following Lemma shows that \Spec\ satisfies the first condition of robust cleanness.

\begin{lemma}
\label{lemma:testing:RefSatisfiesRule1}
Let \Spec\ be constructed from contract \Contract. Then, for all $\sigma, \sigma' \in \traces_\omega(\Spec)$, if $\sigma \in \traces_\omega(\qStd)$, it holds that for all $k \geq 0$ such that $d_\Inputs(\mapInp{\sigma[..j]}, \mapInp{\sigma'[..j]}) \leq \inpbound$
  for all $j\leq k$, there exists
  $\sigma''\in\traces_\omega(\Spec)$ such that
  $\mapInp{\sigma'}=\mapInp{\sigma''}$ and
  $d_\qOutputs(\mapOut{\sigma[k]}, \mapOut{\sigma''[k]}) \leq \outpbound$.
\end{lemma}
\begin{proof}
Let $\sigma \in \traces_\omega(\qStd)$, $\sigma' \in \traces_\omega(\Spec)$ and $k \geq 0$. Suppose that  for all  $j\leq k$ it holds that $d_\Inputs(\mapInp{\sigma[..j]}, \mapInp{\sigma'[..j]}) \leq \inpbound$ (otherwise, the lemma holds trivially). 
From Assumption~\ref{assume:inputs-satisfiable} we know that input $\mapInp{\sigma'}$ is satisfiable and, hence, 
we get from Def.~\ref{def:satisfiableContract} for $\sigma$ and $k$ an implementation $\LTS$ satisfying Def.~\ref{def:ed-clean:LTS}.\ref{def:ed-clean:LTS:ii} 
for which there is $\sigma'' \in \traces_\omega(\LTS_\quiescence)$ with $\mapInp{\sigma''} = \mapInp{\sigma'}$ and $d_\qOutputs(\mapOut{\sigma''[k]}, \mapOut{\sigma[k]}) \leq \outpbound$. 
From Lemma~\ref{lemma:spec-is-largest-clean2} we get that $\traces_\omega(\LTS_\quiescence) \subseteq \traces_\omega(\Spec)$, so we know that $\sigma'' \in \traces_\omega(\Spec)$ 
Hence, $\sigma''$ is the desired trace to conclude the proof.
\qed
\end{proof}

During construction of \Spec\ we added all quiescence transitions that are necessary with regards to the definition of the quiescence closure. This is shown in the following Lemma.

\begin{lemma}
\label{lemma:R-is-quiescence-closed}
Let \Spec\ be constructed from \Contract. Then, the quiescence closure $\Spec_\quiescence$ of \Spec\ is exactly \Spec.
\end{lemma}
\begin{proof}
We have to show that for every state $\sigma \in \traces_*(\Spec)$, there is a transition $\sigma \xrightarrow{o}_\Spec \sigma\cdot o$ in \Spec\ with $o\in\qOutputs$. Let $\sigma_i = \mapInp{\sigma} \cdot (\NoInp)^\omega$ an infinite input trace. We proceed by case-distinction on whether there is a trace $\sigma_S \in \traces_\omega(\qStd)$ such that for all $j \leq |\sigma|+1$ it holds that $d_\Inputs(\sigma_i[..j], \mapInp{\sigma_S[..j]}) \leq \inpbound$. If this is not the case, the premise of Def.~\ref{def:testing:approxSpec} does not hold and hence we get that for all $o\in\qOutputs$ a transition $\sigma \xrightarrow{o}_\Spec \sigma\cdot o$ in \Spec.

If the assumption holds, then we get from assumption~\ref{assume:inputs-satisfiable} and Def.~\ref{def:satisfiableContract} an implementation $\LTS$ and a trace $\sigma'' \in \traces_\omega(\LTS_\quiescence)$ with $\mapInp{\sigma''} = \sigma_i$ and $d_\qOutputs(\mapOut{\sigma''[|\sigma|+1]}, \mapOut{\sigma_S[|\sigma|+1]}) \leq \outpbound$. 
From Lemma~\ref{lemma:spec-is-largest-clean2} we get that $\sigma'' \in \traces_\omega(\Spec)$. 
For this trace to exist it is necessary that there is the transition $\sigma''[..|\sigma|] \xrightarrow{\sigma''[|\sigma|+1]}_\Spec \sigma''[..|\sigma|+1]$ in \Spec. 
Hence, we know (from Def.~\ref{def:testing:approxSpec}) that for every trace $\sigma_S \in \mapInp{\traces_\omega(\mStdPlus)}$ and every $j \leq |\sigma|+1$ for which it holds that $d_\Inputs(\mapInp{\sigma''[..j]}, \mapInp{\sigma_S[..j]}) \leq \inpbound$, there is some $\hat\sigma \in\traces_\omega(\mStdPlus)$ with $\mapInp{\hat\sigma} = \mapInp{\sigma_S}$ and $d_\qOutputs(\sigma''[|\sigma|+1], \mapOut{\hat\sigma[|\sigma|+1]}) \leq \outpbound$.
Since $\mapInp{\sigma''} = \sigma_i$ and in particular $\mapInp{\sigma} \cdot \NoInp = \mapInp{\sigma''[..|\sigma|+1]}$, we have that for every $\sigma_S$ and $j \leq |\sigma|+1$, $d_\Inputs(\mapInp{\sigma''[..j]}, \mapInp{\sigma_S[..j]}) \leq \inpbound \iff d_\Inputs(\mapInp{(\sigma \cdot \sigma''[|\sigma|+1])}, \mapInp{\sigma_S[..j]}) \leq \inpbound$. Hence, we can for every $\sigma_S \in \mapInp{\traces_\omega(\mStdPlus)}$ and $j \leq |\sigma|+1$ with $d_\Inputs(\mapInp{(\sigma \cdot \sigma''[|\sigma|+1])}, \mapInp{\sigma_S[..j]}) \leq \inpbound$, provide a $\hat\sigma \in\traces_\omega(\mStdPlus)$ with $\mapInp{\hat\sigma} = \mapInp{\sigma_S}$ and $d_\qOutputs(\sigma''[|\sigma|+1], \mapOut{\hat\sigma[|\sigma|+1]}) \leq \outpbound$. By Def.~\ref{def:testing:approxSpec} we know that the transition $\sigma[..|\sigma|] \xrightarrow{\sigma''[|\sigma|+1]}_\Spec \sigma[..|\sigma|] \cdot \sigma''[|\sigma|+1]$ exists in \Spec.
Since $\mapInp{\sigma''} = \sigma_i$, we know that $\mapInp{\sigma''}[|\sigma|+1] = \NoInp$ and hence $\sigma''[|\sigma|+1] \in \qOutputs$.
\qed
\end{proof}

We can now prove that \Spec\ is indeed robustly clean. Lemmas~\ref{lemma:testing:RefSatisfiesRule2} and \ref{lemma:testing:RefSatisfiesRule1} show that \Spec satisfies robust cleanness, however, they show each condition w.r.t. \Spec\ instead of $\Spec_\quiescence$. We use Lemma~\ref{lemma:R-is-quiescence-closed} to close this gap.

\SCleanImpliesRClean*
\begin{proof}
  With Lemma~\ref{lemma:testing:RefSatisfiesRule1} and \ref{lemma:R-is-quiescence-closed} we get that \Spec\ satisfies the first condition of Def.~\ref{def:ed-clean:LTS}. From Lemma~\ref{lemma:testing:RefSatisfiesRule2} and \ref{lemma:R-is-quiescence-closed} we conclude that \Spec\ also satisfies the second condition.
  \qed
\end{proof}

\begin{lemma}
\label{lemma:R-is-IOTS}
Let \Contract\ be a contract with standard \Std\ and let \Spec\ be constructed from \Contract. Then \Spec\ is input-enabled.
\end{lemma}
\begin{proof}
We have to show that for any trace $\sigma\in\traces_\omega(\Spec)$ it holds for every $i\in\Inputs$ that $\sigma\cdot i\in\traces_\omega(\Spec)$, i.e., there is a transition $\sigma \xrightarrow{i}_\Spec \sigma \cdot i$ in \Spec.
To have this transition, we must satisfy the premise of Def.~\ref{def:testing:approxSpec}.
Let $\sigma_i\in\mapInp{traces_*(\qStd)}$ and accordingly $\sigma_S\in\traces_*(\qStd)$ a trace with $\mapInp{\sigma_S} = \sigma_i$.
Assume that $\forall j \leq |\sigma|+1$ it holds that $d_\Inputs(\mapInp{(\sigma\cdot i)[..j]}, \sigma_i[..j]) \leq \inpbound$ (otherwise the lemma holds trivially).
We pick $\sigma_S$ for the existential quantifier. 
By definition $\mapInp{\sigma_S} = \sigma_i$, so it suffices to show that $d_\qOutputs(\mapOut{i}, \mapOut{\sigma_S[|\sigma|+1]}) = d_\qOutputs(\NoOutp, \mapOut{\sigma_S[|\sigma|+1]}) \leq \outpbound$.
We continue by case distinction of whether $\sigma_S[|\sigma|+1]\in\Inputs$.
If this is the case, we are immediately done, because $d_\qOutputs(\NoOutp, \NoOutp) = 0 \leq \outpbound$.

If instead $\sigma_S[|\sigma|+1]\in\qOutputs$, we use satisfiability of \Contract\ (Assumption~\ref{assume:inputs-satisfiable}) with $\mapInp{(\sigma\cdot i)}$ for $\sigma_i$, $\sigma_S$ for $\sigma_S$ and $k=|\sigma|+1$. We get some implementation $\LTS$ satisfying Def.~\ref{def:ed-clean:LTS}.\ref{def:ed-clean:LTS:ii} and a trace $\hat\sigma\in\traces_\omega(\LTS_\quiescence)$ with $\mapInp{\hat{\sigma}} = \mapInp{(\sigma\cdot i)}$ and $d_\qOutputs(\mapOut{\hat\sigma[|\sigma|+1]}, \mapOut{\sigma_S[|\sigma|+1]}) \leq \outpbound$. From Lemma~\ref{lemma:spec-is-largest-clean2} we get that $\hat\sigma\in\traces_\omega(\Spec)$. We get from $d_\qOutputs(\mapOut{\hat\sigma[|\sigma|+1]}, \mapOut{\sigma_S[|\sigma|+1]}) \leq \outpbound$ and $\mapInp{\hat{\sigma}} = \mapInp{(\sigma\cdot i)}$ that $d_\qOutputs(\NoOutp, \mapOut{\sigma_S[|\sigma|+1]}) \leq \outpbound$, which concludes the proof.
\qed
\end{proof}

The second property we want to show for \Spec\ is that it is the largest implementation within \Contract.

\RefIsLargestImplementation*
\begin{proof}
We have to show that \Spec\ is an IOTS, it is robustly clean and that for every implementation $\LTS$ that is robustly clean w.r.t. \Contract, it holds that $\traces_\omega(\LTS_\quiescence) \subseteq \traces_\omega(\Spec)$, where \Spec\ is constructed from \Contract.

Lemma~\ref{lemma:R-is-IOTS} gives us that \Spec\ is input-enabled and hence, according to Def.~\ref{def:LTS}, \Spec\ is an IOTS.
We know that \Spec\ is robustly clean from Theorem~\ref{thm:SCleanImpliesRClean}.
From robust cleanness of $\LTS$ we get that $\LTS$ satisfies Def.~\ref{def:ed-clean:LTS}.\ref{def:ed-clean:LTS:ii}. The theorem follows with Lemma~\ref{lemma:spec-is-largest-clean2}.
\qed
\end{proof}

\subsection{Soundness of Algorithm \DT}

In the main paper we presented an algorithm \DT\ to conduct doping tests. We will prove here, that for every robustly clean \SUT\ it indeed holds that $\SUT\ioco\Spec$ and we show soundness and exhaustiveness of the algorithm w.r.t. to this $\ioco$ relation. This leads to the soundness of \DT\ w.r.t. robust cleanness. We then show, that also the bounded version of \DT\ is sound w.r.t. robust cleanness.

We first establish a Lemma that shows an important relation between traces of a system and the outputs of a set of states that can be reached by a certain trace.

\begin{lemma}
\label{lemma:out-extends-trace}
Let $\LTS$ be a LTS, $\sigma \in \traces_*(\LTS_\quiescence)$ a suspension trace of $\LTS$ and $o$ an output. Then, $o \in \out(\LTS_\quiescence \after \sigma)$ if and only if $\sigma\cdot o \in \traces_*(\LTS_\quiescence)$.
\end{lemma}
\begin{proof}
By definition, $o\in\out(\LTS_\quiescence \after \sigma)$ if and only if there is some $q \in \LTS_\quiescence \after \sigma$ for which there is some $q'$ and a transition $q \xrightarrow{o} q'$. This holds if and only if there is a path $p\in\paths_*(\LTS_\quiescence)$ with $\trace(p)=\sigma$, $\last(p)=q$ and $q \xrightarrow{o} q'$. Equivalently, there can be path $p'\in\paths_*(\LTS_\quiescence)$ with $\trace(p') = \sigma\cdot o$, which is the case if and only if $\sigma\cdot o \in \traces_*(\LTS_\quiescence)$.
\qed
\end{proof}

Next, we prove that \Spec\ is a suitable specification for the \ioco\ relation.

\cleanImpliesIoco*
\begin{proof}
We have to show that for all $\sigma \in \traces_*(\Spec_\quiescence)$ it holds that $\out(\SUT_\quiescence \after \sigma) \subseteq \out(\Spec_\quiescence \after \sigma)$. From Lemma~\ref{lemma:R-is-quiescence-closed} we know that $\sigma \in \traces_*(\Spec)$. If $\out(\SUT_\quiescence \after \sigma) = \emptyset$ the theorem trivially holds. Otherwise, there is some $o \in \out(\SUT_\quiescence \after \sigma)$. From Lemma~\ref{lemma:out-extends-trace} we get that $\sigma\cdot o \in \traces_*(\SUT_\quiescence)$. Since every state in $\SUT_\quiescence$ has either an outgoing output or quiescence transition, we know how to extend $\sigma\cdot o$ to an infinite trace $\sigma'\in\traces_\omega(\SUT_\quiescence)$ with $\sigma'[..|\sigma|+1] = \sigma\cdot o$. Since \SUT\ is robustly clean w.r.t. \Contract\ and because \Spec\ is the largest implementation within \Contract, we get with Thm.~\ref{thm:testing:RefIsLargestImplementation} that $\sigma''\in\traces_\omega(\Spec)$. Then, we can conclude that $\sigma\cdot o\in\traces_*(\Spec)$ and $\sigma\cdot o\in\traces_*(\Spec_\quiescence)$ with Lemma~\ref{lemma:R-is-quiescence-closed}. Finally, Lemma~\ref{lemma:out-extends-trace} gives us that $o\in\out(\Spec_\quiescence \after \sigma)$.
\qed
\end{proof}

\removed{
\hrmkSB{This has no essentially been shown in the previous section}
\begin{lemma} \label{lemma:testing:SUTSubsetRef}
Let \Contract\ be a contract with standard \Std\ and \Spec\ be constructed from \Contract. Let \SUT\ be an implementation that is robustly clean w.r.t. \Contract\ and with $\qStd \subseteq \SUT_\quiescence$. Then, $\traces_\omega(\SUT) \subseteq \traces_\omega(\Spec)$.
\end{lemma}

\begin{proof}
  By contradiction, suppose $\sigma \in \traces_\omega(\SUT)$ but
  $\sigma \not\in \traces_\omega(\Spec)$.
  \begin{sloppypar}
  Because $\Spec$ is deterministic, we know that the reason why
  $\sigma$ is no trace of $\Spec$ appears in a finite prefix of
  $\sigma$.
  Let $k$ be the smallest number, for which
  $\sigma[..k] \not\in \traces_*(\Spec)$.
  From $\sigma[..k] \not\in \traces_*(\Spec)$ it follows that there is
  some trace $\sigma_i \in \mapInp{\traces_\omega(\StdPlus)}$ with
  $d_\Inputs(\mapInp{\sigma[..j]}, \sigma_i[..j]) \leq \inpbound$, for
  all $j \leq k$, for which there is no trace
  $\sigma_S \in \traces_\omega(\StdPlus)$ with input $\sigma_i$ for
  which it holds that
  $d_\Outputs(\mapOut{\sigma[k]}, \mapOut{\sigma_S[k]}) \leq \outpbound$.
    \end{sloppypar}

  Let $\sigma_! \in \traces_\omega(\StdPlus)$ be the witness for
  $\sigma_i$, i.e. $\sigma_i = \mapInp{\sigma_!}$.  From
  Def.~\ref{def:ed-clean:LTS} of $\SUT$ for $\sigma$, $\sigma_!$ and
  $k$ and from
  $\forall j \leq k: d_\Inputs(\mapInp{\sigma[..j]}, \sigma_i[..j]) \leq \inpbound$,
  we get a trace $\sigma'' \in \traces_\omega(\SUT)$ with
  $\mapInp{\sigma''} = \sigma_i$ and
  $d_\Outputs(\mapOut{\sigma[k]}, \mapOut{\sigma''[k]}) \leq \outpbound$.
  According to Def. \ref{def:LTS:Std} and because
  $\sigma_i \in \mapInp{\traces_\omega(\StdPlus)}$, it follows that
  $\sigma'' \in \traces_\omega(\StdPlus)$.  This is a contradiction to
  the non-existence of $\sigma_S$ above.
  \qed
\end{proof}
}

In order to prove that Algorithm~\ref{algo:dynamictest} conducts tests according to \ioco\ \Spec, we will recapitulate the original algorithm of the model-based testing theory~\cite{DBLP:journals/cn/Tretmans96,DBLP:conf/fortest/Tretmans08}.


%
First we will consider a distinguished label
$\theta\notin\Inputs\cup\Outputs\cup\{\quiescence\}$ which is intended
to make observable quiescent states. $\theta$ could be implemented via a
 timeout mechanism set appropriately. We refer
to~\cite{DBLP:journals/cn/Tretmans96,DBLP:conf/fortest/Tretmans08} for
further explanations.

\denseparagraph{Test cases.} We work with a smooth and concise notation for test cases, for which we shall use basic process
algebraic notation~\cite{Milner89}.  A \emph{process} is a term
defined in the language given by
$\textstyle p \Coloneqq \nsum_{i\in I} a_i.p_i \mid A$, 
where $I$
is an index set, each $a_i$ is a label, and each $p_i$ is a process,
and $A$ belongs to a set of constants called \emph{process names}
which in turn can be defined by equations of the form $A:=p$.
%
%
We write $\nsum_{i\in I_1} a_i.p_i + \nsum_{i\in I_2} a_i.p_i$ for
$\nsum_{i\in I_1\cup I_2} a_i.p_i$.
A process has semantics in terms of LTS in the usual way: the set of
states is the set of all possible processes and the transitions are
defined according to the  following rules
\[
  \begin{array}{c}\nsum_{i\in I} a_i.p_i \xrightarrow{a_i}p_i\\[2em]\end{array}\qquad\qquad
  \infer{A\xrightarrow{a}p'}{p\xrightarrow{a}p'}\ {\begin{array}{c}A:=p\\[2em]\end{array}}\vspace{-1em}
\]

A \emph{test case} $\test$ for an implementation with inputs in $\Inputs$
and outputs in $\Outputs$ is defined as a deterministic LTS with
the following restrictions:
\begin{inparaenum}[(i)]
\item%
  from $\test$, any of the special processes $\pass$ and $\fail$ can be
  reached, where $\pass\neq\fail$, and they are defined by
  $\pass:=\nsum\{a.\pass\mid a\in\Outputs\cup\{\theta\}\}$
  and
  $\fail:=\nsum\{a.\fail\mid a\in\Outputs\cup\{\theta\}\}$,
\item%
  $t$ has no reachable cycles except those of $\pass$ and $\fail$, and
\item%
  for any state $q$ reachable from $\test$, the set $\{a\mid
  q\xrightarrow{a}q'\}$ contains the whole set $\Outputs$ of
  outputs, and also contains at most one input or $\theta$ (but not
  both).
\end{inparaenum}
A \emph{test suite} is a set of test cases and a
\emph{test run} of a test case $\test$ with an IUT $\SUT$ is an
experiment where the test case supplies inputs to the IUT while
observing the outputs of the IUT or the absence of
them~\cite{DBLP:conf/fortest/Tretmans08}.  This can be described with
a particular form of parallel composition according to the following
rules:
\begin{gather*}
  \!\!\!%
  \infer{%
    q \| p \xrightarrow{a} q' \| p'
  }{
    q \xrightarrow{a} q' & p \xrightarrow{a} p'
  }\ {\begin{array}{c}a \in \Inputs \cup \Outputs\\[2em]\end{array}}
  \qquad\quad
  \infer{%
    q \| p \xrightarrow{\theta} q' \| p
  }{
    q \xrightarrow{\theta} q' & p \xrightarrow{\delta} p'
  }\vspace{-1em}
\end{gather*}
Let $s_0$ be the initial state of $\SUT$ and $q_0$ be the initial
state of $\test$.  The IUT $\SUT$ passes the test case $\test$,
notation $\SUT\passes\test$, if and only if there is no state $s$ such
that a state $\fail \| s$ is reachable from $q_0\|s_0$.  Given a 
test suite 
$T$, we write $\SUT\passes T$ whenever $\SUT\passes\test$ for
all $\test\in T$.

\denseparagraph{Basic test generation.} Tretmans~\cite{DBLP:journals/cn/Tretmans96,DBLP:conf/fortest/Tretmans08} proposed the following recursive and non-deterministic generation algorithm, 
in which ${S \after a} \coloneqq \{s'\mid \exists s\in S: s\xrightarrow{a}s'\}$.

\medskip

\noindent
$\TG(S)\coloneqq$ choose non-deterministically one of the following
processes:
\begin{compactenum}
\item  \label{def:algoTr:1}
$\pass$
\item \label{def:algoTr:2}
$\phantom{+ \ }i;t_i \quad$ where $i {\in} \Inputs$, $s \after i \neq \emptyset$ and $t_i {\in} \TG(S \after i)$\\
$ + \ \nsum \{ o ; \fail \mid {o \in \Outputs} \land {o \notin \out(s)}\}$\\
$ + \ \nsum \{ o_j ; t_{o_j} \mid {o_j \in \Outputs} \land {o_j \in \out(s)}\}$\\
\mbox{}\hspace{6.5em} where for each $o_j$, $t_{o_j} \in \TG(S \after o_j)$
\item  \label{def:algoTr:3}
$\phantom{+ \ }\nsum \{ o; \fail \mid {o \in \Outputs} \land {o \notin \out(s)}\}$\\
$ + \ \nsum \{ \theta ; \fail \mid \quiescence \notin \out(s)\}$\\
$ + \ \nsum \{ o_j ; t_{o_j} \mid {o_j \in \Outputs} \land {o_j \in \out(s)}\}$\\
\mbox{}\hspace{6.5em} where for each $o_j$, $t_{o_j} \in \TG(S \after o_j)$\\
$ +\ \nsum \{ \theta ; t_\theta \mid \quiescence \in \out(s)\}\quad$ where $t_\theta \in \TG(S \after \quiescence)$
\end{compactenum}

\medskip
\noindent
Given a specification $\textit{Spec}$ with initial state $s_0$,
$\TG(\{s_0\})$ generates a \emph{test suite} for $\textit{Spec}$. A test suite is a set of test cases.  
The first possible option in the algorithm states that at any moment the
test process can stop indicating that the execution up to this point
has been satisfactory.  The second option may exercise input $i$ and
continue with test $t_i$. Alternatively it can accept any possible
output. If the output is not included in the specification, the test
fails. If instead it is considered, it is accepted and it continues
with the testing process.  The third option is similar to the previous
one only that it considers the possibility of quiescence instead of
inputs: if this is not a quiescent state but a timeout is produced
without observing outputs (label $\theta$), the test fails; if instead
it is quiescent and a timeout is produced, the test continues with the
selected execution. The important property of this algorithm is that
$\SUT \ioco \textit{Spec}$ if and only if
$\SUT\passes\TG(\{s_0\})$~\cite{DBLP:journals/cn/Tretmans96,DBLP:conf/fortest/Tretmans08}.

\denseparagraph{Complete doping test generation} If $\Contract$ is a satisfiable contract and \Spec\ constructed from \Contract\ with $\epsilon$ being its initial state, then $\TG(\{\epsilon\})$ could be used right away to generate a test suite for $\Spec$, and used to doping test some IUT $\SUT$.
However, $\Spec$ itself cannot be constructed beforehand since it is an infinite object
that requires infinite behaviour to be constructed. We instead propose an algorithm that generates appropriate tests from the quiescence closed standard LTS $\Std_\quiescence$ which is a finite object. 
For this, we use the oracle $\clean$ from the main paper (see eq. (\ref{eq:algo:clean-approx})).

The resulting doping test generation algorithm $\DTG$ is similar to   $\TG$ with
the variation that the outputs that are included in the test come from $\clean$
  rather than from $\out$.

\medskip

\noindent
$\DTG(h)\coloneqq$ 
choose non-deterministically one of the following processes:
\begin{compactenum}
\item \label{def:algoA:1}
$\pass$
\item \label{def:algoA:2}
$\phantom{+ \ }i;t_i \quad$ where $i \in \Inputs$ and $t_i \in \DTG(h \cdot i)$\\
$+ \ \nsum \{ o ; \fail \ | \ o \in \Outputs \land o \notin \clean(h)\}$\\
$+ \ \nsum \{ o_j ; t_{o_j} \ | \ o_j \in \Outputs \land o_j \in \clean(h)\}$\\
\mbox{}\hspace{7.5em} where for each $o_j$, $t_{o_j} \in \DTG(h \cdot o_j)$
\item \label{def:algoA:3}$\nsum \{ o ; \fail \ | \ o \in \Outputs \land o \notin \clean(h)\}$\\
$+ \ \nsum \{ \theta ; \fail \ | \ \quiescence \notin \clean(h)\}$\\
$+ \ \nsum \{ o_j ; t_{o_j} \ | \ o_j \in \Outputs \land o_j \in \clean(h)\}$\\
\mbox{}\hspace{7.5em} where for each $o_j$, $t_{o_j} \in \DTG(h \cdot o_j)$\\
$+ \ \nsum \{ \theta ; t_\theta \ | \ \quiescence \in \clean(h)\}\quad$ where $t_\theta \in \DTG(h \cdot \delta)$
\end{compactenum}
\medskip
\noindent Notice that the input to $\DTG$ is a trace rather than a set of states
as in $\TG$.  Such a trace is, however, the current state of the
reference specification $\Spec$ and, since $\Spec$ is deterministic,
the state $h\cdot a$ is precisely the only possible state of $h \after
a$.
%

In the following proofs, we will deal with transitions of multiple actions. We introduce a handy notation $p \xRightarrow{a\cdot b\cdot c} q$ to express that the system can perform the sequence of actions $\sigma = a\cdot b\cdot c$ in state $p$ to get to state $q$, i.e., $p \xRightarrow{a_1\cdot ... \cdot a_n} q$ if and only if there are states $p_1, \dots, p_n$ such that $p \xrightarrow{a_1} p_1 \xrightarrow{a_2} \cdots \xrightarrow{a_n} p_n = q$.

Lemma~\ref{lemma:equiCleanOut} shows that $\clean$ and $\out$ are equivalent. Notice that $\clean$ works on $\qStd$ and $\out$ operates on \Spec. For ease of reading the proof, we annotate the LTS appropriately, i.e., $\clean^{(\Std_\quiescence)}$ and $\out^{(\Spec)}$.

\begin{lemma} \label{lemma:equiCleanOut}
Let \Contract\ be a contract with standard \Std\ and \Spec\ constructed from \Contract.
%
  For all $h \in (\Inputs \cup \Outputs \cup \{\quiescence\})^*$,
  $\clean^{(\Std_\quiescence)}(h) = \out^{(\Spec)}(\{h\})$.
\end{lemma}

\begin{proof}
  Let $o \in \clean_{\Std_\quiescence}(h)$.
  Hence, we know that $o \in \Outputs \cup \{\quiescence\}$.
  From eq.~(\ref{eq:algo:clean}),
  we get that for all $\sigma_i \in \mapInp{\traces_{\omega}(\mStdPlus)}$,
  if $(\forall j \leq |h|+1: d_\Inputs(\mapInp{\sigma_i[..j]}, \mapInp{(h \cdot o)[..j]}) \leq \inpbound)$
  then there exists $\sigma \in \traces_{\omega}(\mStdPlus)$ such that
  $\mapInp{\sigma} = \mapInp{\sigma_i}$ and
  $d_\Outputs(o, \mapOut{\sigma[|h|+1]}) \leq \outpbound$.
  This is equivalent to the condition on the rule of
  Def.~\ref{def:testing:approxSpec} which proves transition $h
  \xrightarrow{o}_\Spec h \cdot o$. In turns, such transition exists
  if and only if $o \in \out^{(\Spec)}(\{h\})$.
  \qed
\end{proof}

In the following we will proof that $\SUT\ioco\Spec$ if and only if $\SUT\passes\DTG(\epsilon)$, which is essentially the completeness property from~\cite{DBLP:journals/cn/Tretmans96}. We do this by first proving two lemmas. First, Lemma~\ref{lemma:recursiveDTGReachable} shows that for every trace in \Spec, there is a test $t\in\DTG(\epsilon)$ that follows this trace. Second, we show in Lemma~\ref{lemma:testTraceIsRTrace} that, if a call of $\DTG(\epsilon)$ eventually recursively calls $\DTG(\sigma)$, then $\sigma$ is a trace in \Spec. 

\begin{lemma} \label{lemma:recursiveDTGReachable}
  For all $\sigma \in \traces_*(\Spec)$ and all $t' \in \DTG(\sigma)$
  there exists a test $t \in \DTG(\epsilon)$ such that $t
  \xRightarrow{\sigma\lift} t'$. Here, $\sigma \lift$ is the same trace
  as $\sigma$ with all occurrences of $\quiescence$ replaced by
  $\theta$.
\end{lemma}

\remarkPRD{I checked the proof and is serviceable.  So it stays like
  this by the time being.  However it can be improve in clarity and
  elegance.}

\begin{proof}
Let $k = |\sigma|$. Proof by induction on $k$.

If $k=0$, then $t' \Rightarrow t'$. If $k>0$, then we know that for any $t'' \in \DTG(\sigma[..k-1])$, there is a $t \in \DTG(\epsilon)$ such that $t \xRightarrow{\sigma[..k-1]\lift} t''$.

If $\sigma[k] \in \Inputs$, then let $T$ be the set of test cases, that $\DTG(\sigma[..k-1])$ can produce for choice \ref{def:algoA:2}, i.e. $T = \{i. t_i + \nsum\{o_j. t_{o_j} \mid o_j \in \Outputs \land o_j \in \clean(\sigma[..k-1]) \} + \nsum \{o_\ell. \fail \mid o_\ell \in \Outputs \land o_\ell \not\in \clean(\sigma[..k-1])\} \mid i \in \Inputs \land t_i \in \DTG(\sigma[..k-1] \cdot i) \land \forall o_j: t_{o_j} \in \DTG(\sigma[..k-1]\cdot o_j) \}$. Let $T' \subset T$ be the subset of test cases $T$ where $i$ is instantiated by $\sigma[k]$ and $t_i$ by $t'$. There is a $t'' \in T'$ with $t'' \xRightarrow{\sigma[k]\lift} t'$. From the inductive hypothesis we get $t \in \DTG(\epsilon)$ with $t \xRightarrow{\sigma[..k-1]\lift} t''$, so we have $t \xRightarrow{\sigma \lift} t'$.

If $\sigma[k] \in \Outputs$, we get from $\sigma \in \traces_*(\Spec)$ and Lemmas~\ref{lemma:R-is-quiescence-closed} and~\ref{lemma:out-extends-trace} that $\sigma[k] \in \out(\Spec \after \sigma[..k-1])$, which is equivalent to $\sigma[k] \in \out(\{\sigma[..k-1]\})$, since $\Spec$ is deterministic and with Lemma \ref{lemma:testing:RefTraceIsState}. From Lemma \ref{lemma:equiCleanOut} we get that $\sigma[k] \in \clean(\sigma[..k-1])$. Let $T$ be the set of test cases, that $\DTG(\sigma[..k-1])$ can produce for choice \ref{def:algoA:2} (as above). Let $T \supset T' = \{i. t_i + \nsum\{o_j. t_{o_j} \mid o_j \in \Outputs \land o_j \in \clean(\sigma[..k-1]) \} + \nsum \{o_\ell. \fail \mid o_\ell \in \Outputs \land o_\ell \not\in \clean(\sigma[..k-1])\} \mid i \in \Inputs \land t_i \in \DTG(\sigma[..k-1] \cdot i) \land \forall o_j: 
t_{o_j} \in \DTG(\sigma[..k-1]\cdot o_j) \}$. There is a $t'' \in T'$ with $t'' \xRightarrow{\sigma[k]\lift} t'$. From the inductive hypothesis we get $t \in \DTG(\epsilon)$ with $t \xRightarrow{\sigma[..k-1]\lift} t''$, so we have $t \xRightarrow{\sigma \lift} t'$.

If $\sigma[k] = \quiescence$, we get from $\sigma \in \traces_*(\Spec)$ and Lemmas~\ref{lemma:R-is-quiescence-closed} and~\ref{lemma:out-extends-trace}  that $\quiescence \in \out(\Spec \after \sigma[..k-1])$, which is equivalent to $\quiescence \in \out(\{\sigma[..k-1]\})$, since $\Spec$ is deterministic and with Lemma \ref{lemma:testing:RefTraceIsState}. From Lemma \ref{lemma:equiCleanOut} we get that $\sigma[k] \in \clean(\sigma[..k-1])$. Let $T$ be the set of test cases, that $\DTG(\sigma[..k-1])$ can produce for choice \ref{def:algoA:3}, i.e. $T = \{ \nsum \{ o. \fail \mid o \in \Outputs \land o \not\in \clean(\sigma[..k-1]) \} + \nsum \{ \theta. \fail \mid \quiescence \not\in \clean(\sigma[..k-1]) \} + \nsum \{ o_j. t_{o_j} \mid o_j \in \Outputs \land o_j \in \clean(\sigma[..k-1]) \} + \nsum \{ \theta. t_\theta \mid \quiescence \in \clean(\sigma[..k-1])) \} \mid \forall o_j: t_{o_j} \in \DTG(\sigma[..k-1] \cdot o_j) \land t_\theta \in \DTG(\sigma[..k-1] \cdot \quiescence) \}$. Let $T' \subset T$ be the subset of test cases $T$ where $t_\theta$  is instantiated by  $t'$. There is a $t'' \in T'$ with $t'' \xRightarrow{\sigma[k]\lift} t'$. From the inductive hypothesis we get $t \in \DTG(\epsilon)$ with $t \xRightarrow{\sigma[..k-1]\lift} t''$, so we have $t \xRightarrow{\sigma \lift} t'$.
\qed
\end{proof}

The following Lemma shows, that if a recursive call $\DTG(\sigma)$ is reachable from $\DTG$ for the empty history, then $\sigma$ is a trace of $\Spec$.

\begin{lemma} \label{lemma:testTraceIsRTrace}
Let $\sigma' \in (\Inputs \cup \Outputs \cup \{\quiescence\})^*$ be any trace, $\sigma \in \traces_*(\Spec)$ a trace in $\Spec$ and $t \in \DTG(\sigma)$ a test case generated for history $\sigma$. If $t \xRightarrow{\sigma' \lift} \nsum_{i \in I} a_i. t_i$ and if there is $i \in I$ for which $t_i \not\in \{\fail, \pass\}$, then $\sigma \cdot \sigma' \in \traces_*(\Spec)$.
\end{lemma}

\begin{proof}
By induction on $|\sigma'|$. If $\sigma' = \epsilon$, then $\sigma \cdot \sigma' \in \traces_*(\Spec)$ follows trivially. For $|\sigma| > 0$, we first observe that $t \xrightarrow{\sigma'[1]\lift} t' \xRightarrow{\sigma'[2..]\lift} \nsum_{i \in \SUT} a_i. t_i$. Notice that it must be that $t' = \nsum_{j \in J} b_j. t'_j$ and $\exists j \in J: t_j \not\in \{\fail, \pass\}$, because once state $\fail$ or $\pass$ is reached, the LTS is trapped in this state and would violate $\exists i \in I: t_i \not\in \{\fail, \pass\}$. For the same reason, $t = \nsum_{\ell \in L} c_\ell. t''_\ell$ and $\exists \ell \in L: t_\ell \not\in \{\fail, \pass\}$.

The transition $t \xrightarrow{\sigma'[1]\lift} t'$ proves that for some $\ell \in L$, $c_\ell.t''_\ell = \sigma'[1]. t'$. Considering all possible choices how $\sigma'[1]. t'$ was added to the sum in $t$, we get that it must be that $t' \in \DTG(\sigma \cdot \sigma'[1])$ and that either $\sigma'[1] \in \clean(\sigma)$ or $\sigma'[1]$ is an input. If $\sigma'[1]$ is an input, we know that $\sigma \cdot \sigma'[1] \in \traces_*(\Spec)$ because \Spec\ must be input enabled. If $\sigma'[1] \in \clean(\sigma)$, we get from Lemma \ref{lemma:equiCleanOut} that $\sigma'[1] \in \out(\{\sigma\})$, so $\sigma \cdot \sigma'[1] \in \traces_*(\Spec)$ with Lemma~\ref{lemma:out-extends-trace}.

From the inductive hypothesis we get for $\sigma'[2..]$ and $t'$ that $\sigma \cdot \sigma'[1] \cdot \sigma'[2..] \in \traces_*(\Spec)$ and hence $\sigma \cdot \sigma' \in \traces_*(\Spec)$.
\qed
\end{proof}

The next lemma states that every robustly clean implementation
passes the test suite generated by the algorithm $\DTG$, i.e., it shows completeness of the test generation algorithm.

\begin{restatable}{lemma}{IocoPassesDTG}
\label{lemma:IocoPassesDTG}
Let $\Contract$ be a contract with standard $\Std$. Let $\SUT$ be an implementation with $\Std_\quiescence \subseteq \SUT_\quiescence$
and let $\Spec$ be the largest implementation within $\Contract$. Then, $\SUT \ioco \Spec$ if and only if  $\SUT \passes \DTG(\epsilon)$.
\end{restatable}
\begin{proof}
We prove both directions by contraposition, i.e., $\neg(\SUT\ioco\Spec) \iff \exists t \in \DTG(\epsilon): \neg(\SUT \passes t)$.

Unrolling $\ioco$ according to (\ref{eq:doping:test}) and unrolling the reachability of \passes\!\!, gives us the equivalent claim $\exists \sigma \in \traces_*(\Spec_\quiescence): \exists o \in \qOutputs: o\ \in \out(\SUT_\quiescence\after\sigma) \land o \not\in \out(\Spec_\quiescence\after\sigma) \iff \exists t\ \in \DTG(\epsilon): \exists \sigma \in (\Inputs \cup \Outputs \cup \{\quiescence\})^*: \exists s_1, s_2, t': t \parallel s_0 \xRightarrow{\sigma[..|\sigma|-1] \lift} t' \parallel s_2 \xrightarrow{\sigma[|\sigma|] \lift} \fail || s_1$, where $s_0$ is the initial state of $\SUT$.

We continue to prove the two directions separately.

\begin{itemize}
\item $\Rightarrow$:

From $o \in \out(\SUT_\quiescence\after\sigma)$ and Lemma~\ref{lemma:out-extends-trace}  we get that $(\sigma \cdot o) \in \traces_*(\SUT_\quiescence)$ and that $s_0 \xRightarrow{\sigma} s_2 \xrightarrow{o} s_1$ for some $s_1$ and $s_2$, where $s_0$ is the initial state of $\SUT$. Since $\Spec$ is deterministic and from Lemma \ref{lemma:testing:RefTraceIsState}, we know $\Spec \after \sigma = \{\sigma\}$. From Lemmas~\ref{lemma:R-is-quiescence-closed},~\ref{lemma:equiCleanOut} and $o \not\in \out(\Spec_\quiescence\after\sigma)$, we get that $o \not\in \clean(\sigma)$.

Let $t' = \nsum \mathcal{M} \in \DTG(\sigma)$, where $\DTG$ chose choice \ref{def:algoA:3}. Since $o \not\in \clean(\sigma)$, we know that $o. \fail \in \mathcal{M}$ and that $t' \xrightarrow{o\lift} \fail$.

Since $\sigma \in \traces_*(\Spec_\quiescence)$ and hence $\sigma \in \traces_*(\Spec)$ (Lemma~\ref{lemma:R-is-quiescence-closed}), we get from Lemma \ref{lemma:recursiveDTGReachable}, that there is some $t \in \DTG(\epsilon)$ with $t \xRightarrow{\sigma\lift} t'$.

Hence, by the definition of $\parallel$, we conclude that $t \parallel s_0 \xRightarrow{\sigma\lift} t' \parallel s_2$ from $t \xRightarrow{\sigma\lift} t'$ and $s_0 \xRightarrow{\sigma} s_2$. Also, $t' \xrightarrow{o\lift} \fail$ and $s_2 \xrightarrow{o} s_1$ imply $t' \parallel s_2 \xrightarrow{o\lift} \fail \parallel s_1$.
With $t$ and $\sigma \cdot o$ we prove this direction of the proof.

\item $\Leftarrow$:

W.l.o.g. let $\sigma$ be the shortest trace and $s_1$ and $s_2$ such that $t \parallel s_0 \xRightarrow{\sigma} \fail || s_1$.

Since $t' \parallel s_2 \xrightarrow{\sigma[|\sigma|]\lift} \fail \parallel s_1$ we know from Def. $\parallel$ that $t' \xrightarrow{\sigma[|\sigma|]\lift} \fail$. Since $\sigma$ is the shortest trace reaching $\fail$, $\sigma[|\sigma|]. \fail$ was introduced by $\DTG$ and hence $t' = \nsum_{i\in I} a_i. t_i$ where $a_i.t_i = \sigma[|\sigma|]. \fail$ for some $i$ and there is some $j \in I$ with $t_j \not\in \{\fail,\pass\}$. With Lemma \ref{lemma:testTraceIsRTrace} we can follow that $\sigma[..|\sigma|-1]$ is in $\traces_*(\Spec)$.

We consider all possible cases how $\DTG$ can add $\sigma[|\sigma|]. \fail$ to $t'$. We observe, that in all cases, it is necessary that $\sigma[|\sigma|] \in \Outputs \cup \{\quiescence\}$ and $\sigma[|\sigma|] \not\in \clean(\sigma[..|\sigma|-1])$. With Lemma \ref{lemma:equiCleanOut} we get that $\sigma[|\sigma|] \not\in \out(\{\sigma[..|\sigma|-1]\})$ and since $\Spec$ is deterministic and by Lemma \ref{lemma:testing:RefTraceIsState}, this is equivalent to $\sigma[|\sigma|] \not\in \out(\Spec\after\sigma[..|\sigma|-1])$. 

From $t \parallel s_0 \xRightarrow{\sigma[..|\sigma|-1] \lift} \nsum \mathcal{M} \parallel s_2 \xrightarrow{\sigma[|\sigma|] \lift} \fail || s_1$ we get from the definition of $\parallel$, that $s_0 \xRightarrow{\sigma[..|\sigma|-1]} s_2 \xrightarrow{\sigma[|\sigma|]} s_1$. Hence $\sigma[|\sigma|] \in \out(\SUT_\quiescence \after \sigma[..|\sigma|-1])$.

Now, $\sigma[..|\sigma|-1]$ and $\sigma[|\sigma|]$ prove this implication.
\end{itemize}
\qed
\end{proof}

Notice that $\DTG$ is exactly \DT\ when instantiating $\Omega_\text{case}$ and $\Omega_\Inputs$ by non-determinism. The first case of each algorithm is identical. For the second case, \DT\ does not accept outputs from the IUT. However, there is an additional premise for entering case 2 that allows this case only when no output from the IUT is available. We can check this in the premise, because our algorithm is an online testing algorithm, whereas \DTG\ is a test case generation algorithm -- \DTG\ produces a LTS that must be capable of handling an output in case 2. 
Case 3 is identical in both algorithms, however, we use case 3 to process an output from the IUT if case 2 of the algorithm was chosen but an output is available. This is realised by weakening the premise so to allow this case also if an output from the IUT is available.
The observation that $\DTG$ is exactly \DT\ when instantiating $\Omega_\text{case}$ and $\Omega_\Inputs$ by non-determinism (Assumption~\ref{assume:omega_nondet}) proves the following theorem from the main paper.

\iocoDT*
\begin{proof}
Algorithm~\ref{algo:dynamictest} with $\Omega_\text{case}$ and $\Omega_\Inputs$ being the non-deterministic choice is equivalent to \DTG.
Hence, the theorem follows from Lemma~\ref{lemma:IocoPassesDTG}.
\qed
\end{proof}

\denseparagraph{Bounded-depth doping test generation}
In order to provide practicable testing, we introduced a bounded-length testing algorithm in the main paper. We can modify \DTG\ exactly as we did for \DT: Instead of $\clean$ we use the bounded version $\clean_b$ that can be computed. Since $\clean_b$ only considers finite traces, it conservatively
includes extra outputs thus making tests more permissive.  This is due
to the existential quantifier in the last line of
(\ref{eq:algo:clean-approx}): it may be the case that the $b$-prefix
of some infinite trace satisfies this expression, but no infinite
extension of such prefix in $\Std_\quiescence$ does.
Therefore, we have the following variation of
Lemma~\ref{lemma:equiCleanOut}.

\begin{lemma} \label{lemma:supseteqCleanOut}
Let \Contract\ be a contract with standard \Std\ and \Spec\ constructed from \Contract.
  For all $b>0$, $h \in (\Inputs \cup \Outputs \cup \{\quiescence\})^*$ with $|h| < b$,
  $\clean_b^{(\Std_\quiescence)}(h) \supseteq
  \clean^{(\Std_\quiescence)}(h) = \out^{(\Spec)}(\{h\})$.
\end{lemma}
\begin{proof}
$\clean^{(\Std_\quiescence)}(h) = \out^{(\Spec)}(\{h\})$ is already shown in Lemma~\ref{lemma:equiCleanOut}.
We only have to prove $\clean^{(\Std_\quiescence)}(h) \subseteq \clean_b^{(\Std_\quiescence)}(h)$, i.e., $\forall o\in\qOutputs$ it holds that $o\in\clean(h) \iff o\in\clean_b(h)$.
Let $\mathsf{P}$ be the predicate $\mathsf{P}(\sigma) \coloneqq \forall j \leq |h|+1: d_\Inputs(\mapInp{\sigma[..j]}, \mapInp{(h\cdot o)}) \leq \inpbound$ for all $\sigma$ with $|\sigma| > |h|$. 

Assume $o \in \clean(h)$. We have to show that $o\in\clean_b(h)$. For this, we have to assume an input trace $\sigma_i\in\mapInp{\traces_b(\qStd)}$ with $\mathsf{P}(\sigma_i)$. Let $\hat{\sigma_i} = \sigma_i \cdot \NoInp^\omega$. $\hat{\sigma_i} \in \mapInp{\traces_\omega(\qStd)}$, because every finite trace of \qStd\ can be extended by either quiescence or an output. Since $|h| < b$ and hence $|h\cdot o| \leq b$, $\mathsf{P}(\hat{\sigma_i})$ holds. From $o\in\clean(h)$ we now get that that there is some trace $\hat\sigma\in\traces_\omega(\qStd)$ with $\mapInp{\hat\sigma} = \hat{\sigma_i}$ and $d_\qOutputs(o, \mapOut{\hat\sigma[|h|+1]}) \leq \outpbound$. Let $\sigma = \hat\sigma[..|h|+1]$, then $\mapInp{\sigma} = \sigma_i$ and $d_\qOutputs(o, \mapOut{\sigma[|h|+1]}) \leq \outpbound$. This proves $o\in\clean_b(h)$.
\qed
\end{proof}

As a consequence of Lemma~\ref{lemma:supseteqCleanOut}, we have that
any robustly clean implementation passes the test suit generated by
$\DTG_b$, or, expressed inversely, if an implementation fails a test
generated by $\DTG_b$, then it is doped.  This is stated in the
following lemma.


\begin{restatable}{lemma}{IocoPassesDTGb}
\label{lemma:IocoPassesDTGb}
Let $\Contract$ be a contract with standard $\Std$. Let $\SUT$ be an implementation with $\Std_\quiescence \subseteq \SUT_\quiescence$. Then, if \SUT\ is robustly clean w.r.t. \Contract, $\SUT\passes\DTG_b(\epsilon)$ for every positive integer $b$.
\end{restatable}
\begin{proof}
%
From Theorem~\ref{thm:testing:cleanImpliesIoco}, we know that if $\SUT$ is robustly clean then $\SUT \ioco \Spec$. So, we proceed by contradiction and assume that $\SUT \ioco \Spec$ and $\neg(\SUT \passes \DTG_b(\epsilon))$.

Since $\neg(\SUT \passes \DTG_b(\epsilon))$ it must hold that there is some $t \in \DTG_b(\epsilon)$ and $\sigma \in (\Inputs \cup \Outputs \cup \{\quiescence\})^*$ such that $\exists s': t \parallel s_0 \xRightarrow{\sigma} \fail \parallel s'$, where $s_0$ is the initial state of $\SUT$.
%
%
Hence, there is some $s''$ and $t'$ such that $t \parallel s_0 \xRightarrow{\sigma[..|\sigma|-1]} t' \parallel s'' \xrightarrow{\sigma[|\sigma|]\lift} \fail \parallel s'$.

W.l.o.g. let $\sigma$ be the shortest trace and $s'$ and $s''$ such that $t \parallel s_0 \xRightarrow{\sigma} \fail || s'$.

Since $t' \parallel s'' \xrightarrow{\sigma[|\sigma|]\lift} \fail \parallel s'$ we know from Def. $\parallel$ that $t' \xrightarrow{\sigma[|\sigma|]\lift} \fail$. Since $\sigma$ is the shortest trace reaching $\fail$, $\sigma[|\sigma|]. \fail$ was introduced by $\DTG_b$ and hence $t' = \nsum_{i\in I} a_i. t_i$ where $a_i.t_i = \sigma[|\sigma|]. \fail$ for some $i$ and there is some $j \in I$ with $t_j \not\in \{\fail,\pass\}$. 
%
We consider all possible cases how $\DTG_b$ can add $\sigma[|\sigma|]. \fail$ to $t'$. We observe, that in all cases, it is necessary that $\sigma[|\sigma|] \in \Outputs \cup \{\quiescence\}$ and $\sigma[|\sigma|] \not\in \clean_b(\sigma[..|\sigma|-1])$. 
With Lemma \ref{lemma:supseteqCleanOut} we get that $\sigma[|\sigma|] \not\in \out(\{\sigma[..|\sigma|-1]\})$ and since $\Spec$ is deterministic and with Lemma \ref{lemma:testing:RefTraceIsState}, this is equivalent to $\sigma[|\sigma|] \not\in \out(\Spec\after\sigma[..|\sigma|-1])$ (*). 

From $\SUT \ioco \Spec$ we get that for any finite trace $\sigma \in\traces_*(\Spec_\quiescence)$ it holds that $\out(\SUT_\quiescence \after \sigma) \subseteq \out(\Spec_\quiescence \after \sigma)$. Since $s_0 \xRightarrow{\sigma} s'$, we know that $\sigma \in \traces_*(\SUT_\quiescence)$. Let $\sigma' \in\traces_\omega(\SUT_\quiescence)$ be an infinite extension of $\sigma$ with $\sigma'[..|\sigma|] = \sigma$. By knowing that $\SUT$ is robustly clean and Theorem~\ref{thm:testing:RefIsLargestImplementation} it follows that $\sigma' \in\traces_\omega(\Spec)$. Hence, $\sigma$ is a finite trace of \Spec\ and from Lemma~\ref{lemma:out-extends-trace} and lemma~\ref{lemma:R-is-quiescence-closed} we get that $\sigma[|\sigma|] \in \out(\Spec\after\sigma[..|\sigma|-1])$, which contradicts (*).
%
%
\qed
\end{proof}

\boundedPassing*
\begin{proof}
$DT_b$ (Algorithm~\ref{algo:dynamictest} with $\clean_b$) with $\Omega_\text{case}$ and $\Omega_\Inputs$ being the non-deterministic choice is equivalent to $\DTG_b$.
Hence, the theorem follows from Lemma~\ref{lemma:IocoPassesDTGb}.
\qed
\end{proof}


\removed{
It turns out that, by construction of $\Spec$, the suspension traces of
$\Spec$ are exactly its finite traces.

\begin{lemma} \label{lemma:tracesAreSTraces}
  If $\Spec$ is a consistent reference specification, then
  $\STraces(\Spec)=\traces_*(\Spec)$.
\end{lemma}


\begin{proof}
  First notice that only the inclusion
  $\STraces(\Spec)\subseteq\traces_*(\Spec)$ is relevant (the other
  holds trivially).  For a proof by contradiction, suppose that there is some $\sigma \in \STraces(\Spec)$ but not in $\traces_*(\Spec)$. There must be some path $p \in \paths_*(\Spec)$ with $\last(p) = \sigma'$, $\sigma' = \sigma[..k]$ for some $k < |\sigma|$ and $\sigma' \not\xrightarrow{\sigma[k+1]}_\Spec \sigma[..k+1]$. According to Def. \ref{def:delta-closure}, it must be that $\sigma' \xrightarrow{\quiescence}_{\Spec_\quiescence} \sigma'$ and hence $\forall o \in \Outputs \cup \{\quiescence\}, \sigma'' \in Q : \sigma' \not\xrightarrow{o}_\Spec \sigma''$. From Def. \ref{def:testing:approxSpec} we get (for any $o$ and in particular for $o = \quiescence$) some $\sigma_i \in \mapInp{\traces_\omega(\StdPlus)}$ and hence $\sigma_! \in \traces_\omega(\StdPlus)$ with $\forall j \leq k+1: d_\Inputs(\mapInp{\sigma''[..j]}, \mapInp{\sigma_![..j]}) \leq \inpbound$.
  
Moreover, we get from $\forall o \in \Outputs \cup \{\quiescence\}, \sigma'' \in Q : \sigma' \not\xrightarrow{o}_\Spec \sigma''$ that $\mathcal{T} = \{ \sigma \in \traces_\omega(\mathcal{R}) \mid \mapInp{\sigma} = \mapInp{\sigma'} \cdot \NoInp \} = \emptyset$. Hence, we know that $\sigma \cdot \quiescence \in \BadInputs$ for $\sigma_!$ and $k+1$, which contradicts the consistency of $\Spec$.
    \qed
  
%

\end{proof}
}

\removed{
Now we are in conditions to prove Theorem \ref{thm:testing:cleanImpliesIoco}.

\begin{proof}[Proof of Theorem 2]
  Let $\sigma \in \STraces(\Spec)$. By definition, $\STraces(\Spec) = \traces_*(\Spec_\quiescence)$, which is equal to $\traces_*(\Spec)$ by Lemma~\ref{lemma:R-is-quiescence-closed}. Hence $\sigma \in \traces_*(\Spec)$. 
  Suppose there
  is some $o\in\out(\SUT \after \sigma)$ (otherwise the theorem
  trivially holds).
  We need to show that $o\in\out(\Spec \after \sigma)$.
  Let $\hat{\sigma} = \sigma \cdot o$.
  By Lemma~\ref{lemma:testing:RefTraceIsState}, for any path
  $p\in\paths_*(\Spec)$ with $\trace(p)=\sigma$,
  we know $\last(p)=\sigma$.
  Hence we have to show that there is a transition $\sigma
  \xrightarrow {o}_\Spec \hat{\sigma}$.

  By Def.~\ref{def:testing:approxSpec} we have to show for some
  $\sigma_i \in \mapInp{\traces_\omega(\StdPlus)}$ that is close
  to the input of $\hat\sigma$ (i.e.\
  $\forall j \leq |\hat{\sigma}|:
  d_\Inputs(\mapInp{\hat{\sigma}[..j]}, \sigma_i[..j]) \leq \inpbound$),
  that there is some trace $\sigma_S \in \traces_\omega(\StdPlus)$
  such that $\mapInp{\sigma_S} = \sigma_i$ and
  $d_\Outputs(o, \mapOut{\sigma_S[|\hat{\sigma}|]}) \leq \outpbound$ (*).  

  Note that, since $\sigma_i$ is in $\mapInp{\traces_\omega(\StdPlus)}$,
  there is some $\sigma^* \in \traces_\omega(\StdPlus)$ with
  $\mapInp{\sigma^*} = \sigma_i$.

  \begin{sloppypar}
  Let $\sigma'$ be an infinite extension of $\hat{\sigma}$ in
  $\SUT$ and notice that for all $j \leq |\hat{\sigma}|$,
  $d_\Inputs(\mapInp{\sigma'[..j]}, \mapInp{\sigma^*[..j]})
  = d_\Inputs(\mapInp{\hat{\sigma}[..j]}, \sigma_i[..j]) \leq \inpbound$.
  Since $\SUT$ is robustly clean, because of the second condition in
  Def.~\ref{def:ed-clean:LTS}, we know that there exists
  $\sigma''\in\traces_\omega(\SUT)$ with
  $\mapInp{\sigma^*} =\mapInp{\sigma''}$ for which
  $d_\Outputs(\mapOut{\hat{\sigma}[..k]}, \mapOut{\sigma''[..k]}) \leq \outpbound$.
  Because $d_\Outputs$ is past-forgetful, this last statement amounts
  to $d_\Outputs(o, \mapOut{\sigma''[..k]}) \leq \outpbound$.
  Since $\sigma^* \in \traces_\omega(\StdPlus)$ and
  $\mapInp{\sigma^*} = \mapInp{\sigma''}$ we know by
  Def.~\ref{def:LTS:Std} that $\sigma'' \in
  \traces_\omega(\StdPlus)$. Then, trace $\sigma''$ is the trace
  $\sigma_S$ that we were looking for in (*).
  \qed
    \end{sloppypar}
\end{proof}
}

\removed{\IocoPassesDTG*

To prepare for the proof of this theorem, the next lemma states that the outputs observable at a state $h$ in
$\Spec$ are precisely those obtained with $\clean(h)$ from the
standard $\Std_\quiescence$. In order to make clear where each
function is applied, we use appropriate superscripts.}

\removed{

This lemma is central to show that, if $\Std_\quiescence$ induces a
consistent reference specification, $\DTG$ generates a test suite that
completely captures the conformance relation $\SUT\passes\Spec$.  It
is in fact used by the following lemma which states that any
trace in $\Spec$ is exercised in any possible way by the test suite
generated by $\DTG$.


Now, we are able to prove Theorem \ref{thm:IocoPassesDTG}.


Next, we are heading for the final Theorem.

\IPassesBoundedDTG*}




\removed{
The following lemma states the soundness of predicate $\BadInputsFin$.

\begin{lemma}
Let $h \in \traces_*(\Spec)$.  If $\BadInputsFin(h)$, then $\Spec$ is inconsistent.
  

%
%
\end{lemma}

\remarkPRD{The proof looks sound, though it can be made more readable (do not spend time on this now)}

\begin{proof}
We have to show, that there is an infinite trace $\sigma \in \traces_\omega(\Spec_\quiescence)$, which is in $\BadInputs$.

Let $p \in \paths_*(\Spec)$ be a path with $\trace(p) = h$. By Lemma \ref{lemma:testing:RefTraceIsState}, we know that $\last(p) = h$.
Let $p^+ \in \paths_\omega(\Spec_\quiescence)$ be an infinite extension of $p$, i.e., $p^+[..|h|] = p$. Note that this extension always exists, even for inconsistent $\Spec$, because the quiescence closure guarantees that in any state of $\Spec$, there is an outgoing transition that is either labelled with an output action or a $\quiescence$. 

If $\clean_b(h) = \emptyset$, then by Lemma \ref{lemma:supseteqCleanOut} we have that $\out(h) = \emptyset$. Hence, by Def. \ref{def:delta-closure}, there is a transition $h \xrightarrow{\quiescence}_{\Spec_\quiescence} h$ in $\Spec_\quiescence$. Note that $p^+ = s_0 \; a_0 \; s_1 \cdots h \; a_{|h|} \; s_{|h|+1} \cdots$. Let $p' = s_0 \; a_0 \; s_1 \cdots h \; \quiescence \; h \; a_{|h|} \; s_{|h|+1} \cdots$ be $p^+$ with the additional transition $h \xrightarrow{\quiescence}_{\Spec_\quiescence} h$ in the middle. Obviously $p' \in \paths_\omega(\Spec_\quiescence)$.

From $h \not\xrightarrow{\quiescence}_\Spec h \cdot \quiescence$ and Def. \ref{def:testing:approxSpec}, we get a trace $\sigma_i \in \mapInp{\traces_\omega(\StdPlus)}$ and hence $\sigma_! \in \traces_\omega(\StdPlus)$ with $\forall j \leq |h|+1: d_\Inputs(\mapInp{(h\cdot\quiescence)[..j]}, \mapInp{\sigma_![..j]}) \leq \inpbound$.

Let $\sigma = \trace(p')$. Notice that $\sigma[..|h|] = h$ and that $\mapInp{\sigma[|h|+1]} = \NoInp$. Since $\out(\{h\}) = \emptyset$, we know that for all $\sigma' \in \traces_\omega(\Spec)$, $\mapInp{\sigma'[|h|+1]} \neq \NoInp$. Hence $\{ \sigma' \in \traces_\omega(\Spec) \mid \mapInp{\sigma'} = \mapInp{\sigma} \} = \emptyset$.

Now, it follows that $\sigma \in \BadInputs$ with $\sigma_!$ and $|h|+1$.


If $\clean_b(h) \neq \emptyset$, then we get from $\BadInputsFin(h)$ a trace $\sigma_S \in \traces_{|h|}(\StdPlus)$ and $k \geq 0$ such that $\forall j \leq k: d_\Inputs(\mapInp{h[..j]}, \mapInp{\sigma_S[..j]}) \leq \inpbound$. Let $h^+ = \trace(p^+)$ and $\sigma_S^+ \in \traces_\omega(\StdPlus)$ be an infinite extension of $\sigma_S$. Since $k \leq |h|$ it holds that $\forall j \leq k: d_\Inputs(\mapInp{h^+[..j]}, \mapInp{\sigma_S^+[..j]}) \leq \inpbound$.

Let $\sigma^+ \in \traces_\omega(\Spec)$ be an arbitrary trace in $\Spec$ with $\mapInp{\sigma^+} = \mapInp{h^+}$ and let $\sigma \in \traces_{|h|}(\Spec)$ be the prefix $\sigma^+[..|h|]$ of $\sigma^+$. Since $|\sigma| = |h|$ we have $\mapInp{\sigma} = \mapInp{h}$. Furthermore, we know from $\BadInputsFin(h)$ that $\forall a \in \clean_b(h[..k-1]) \cup \Inputs: d_\Outputs(\mapOut{a}, \mapOut{\sigma_S[k]}) > \outpbound$ (*).  We continue by case distinction.

\begin{itemize}
\item If $\sigma[k] \in \Inputs$, then it follows from (*) that $d_\Outputs(\mapOut{\sigma[k]}, \mapOut{\sigma_S[k]}) > \outpbound$ and hence, since $k \leq |h|$, $d_\Outputs(\mapOut{\sigma^+[k]}, \mapOut{\sigma_S^+[k]}) > \outpbound$.
\item Otherwise, $\sigma[k] \in \Outputs \cup \{\quiescence\}$ and we show that $\sigma[k] \in \clean_b(h[..k-1])$ as follows. By (\ref{eq:algo:clean-approx}), we have to show that $\forall \sigma_i \in \traces_b(\Std_\quiescence): (\forall j \leq k: d_\Inputs(\mapInp{\sigma_i[..j]}, \mapInp{(h[..k-1] \cdot \sigma[k])[..j]}) \leq \inpbound) \Rightarrow \exists \sigma_S \in \traces_b(\Std_\quiescence): \mapInp{\sigma_S} = \mapInp{\sigma_i} \land d_\Outputs(\sigma[k], \mapOut{\sigma_S[k]}) \leq \outpbound$. Since $|\sigma| = |h|$ we know that $\mapInp{(h[..k-1] \cdot \sigma[k])} = \mapInp{(\sigma[..k-1] \cdot \sigma[k])}$ which gives us $\sigma[k] \in \clean_b(h[..k-1]) \iff \sigma[k] \in \clean_b(\sigma[..k-1])$ after some rewriting of (\ref{eq:algo:clean-approx}). To prove $\sigma[k] \in \clean_b(\sigma[..k-1])$, by Lemma \ref{lemma:supseteqCleanOut} it suffices to show that $\sigma[k] \in \out(\{\sigma[..k-1]\})$, which follows from $\sigma \in \traces_{|h|}(\Spec)$ and $k \leq |h|$.

For $\sigma[k] \in \clean_b(\sigma[..k-1])$ we get from (*) that $d_\Outputs(\mapOut{\sigma[k]}, \mapOut{\sigma_S[k]}) > \outpbound$ and hence, since $k \leq |h|$, $d_\Outputs(\mapOut{\sigma^+[k]}, \mapOut{\sigma_S^+[k]}) > \outpbound$.
\end{itemize}
Now $\sigma_S^+$ and $k$ prove that $h^+ \in \BadInputs$ and hence that $\Spec$ is inconsistent.
\end{proof}
}


\fi

\end{document}